\documentclass[hyperfootnotes=false,11pt,a4paper,margin=1in]{article}

\usepackage[english]{babel}
\usepackage{cite}

\usepackage{setspace}
\usepackage[utf8]{inputenc}
\usepackage{fullpage}
\usepackage{amsmath, amsfonts, amssymb, amsthm}
\usepackage{color}
\usepackage{xcolor}
\definecolor{darkgreen}{rgb}{0,0.5,0}
\definecolor{darkblue}{rgb}{0,0,0.6}

\usepackage[linesnumbered,noresetcount,vlined,ruled,procnumbered]{algorithm2e}
\usepackage{hyperref}
\hypersetup{
    unicode=false,          %
    colorlinks=true,        %
    linkcolor=darkblue,          %
    citecolor=darkgreen,        %
    filecolor=magenta,      %
    urlcolor=blue           %
}
\usepackage[capitalize, nameinlink]{cleveref}

\usepackage{graphicx}
\usepackage{comment}
\usepackage{thm-restate}
\usepackage{bbm}
\usepackage{tikz}
\usepackage{pgfplots}
\usepgfplotslibrary{fillbetween} %
\pgfplotsset{compat=1.16} %
\usetikzlibrary{patterns}

\usepackage[linewidth=1pt]{mdframed}
\usepackage[most]{tcolorbox}
\definecolor{block-gray}{gray}{0.85}
\newtcolorbox{shadequote}{colback=block-gray,grow to right by=-2mm,grow to left by=-2mm,
boxrule=0pt,boxsep=0pt,breakable}
\usepackage{color}
\usepackage[inline]{enumitem}

\usepackage{caption}
\usepackage{cancel}
\usepackage{subcaption}
\usepackage{mathrsfs}
\usepackage{xspace}

\usepackage{mathtools} %
\usepackage{booktabs} %
\usepackage{makecell}%
\usepackage{tabularx}

\SetKwComment{Comment}{$\quad\triangleright$\ }{}
\newcommand{\nosemic}{\renewcommand{\@endalgocfline}{\relax}}%

\usepackage[normalem]{ulem} %

\crefalias{AlgoLine}{line}%

\makeatletter
\let\cref@old@stepcounter\stepcounter
\def\stepcounter#1{%
  \cref@old@stepcounter{#1}%
  \cref@constructprefix{#1}{\cref@result}%
  \@ifundefined{cref@#1@alias}%
    {\def\@tempa{#1}}%
    {\def\@tempa{\csname cref@#1@alias\endcsname}}%
  \protected@edef\cref@currentlabel{%
    [\@tempa][\arabic{#1}][\cref@result]%
    \csname p@#1\endcsname\csname the#1\endcsname}}
\makeatother

\usepackage[titletoc,title]{appendix}

\newtheorem{theorem}{Theorem}[section]
\newtheorem{lemma}[theorem]{Lemma}
\newtheorem{remark}{Remark}
\newtheorem{definition}{Definition}
\newtheorem{proposition}{Proposition}
\newtheorem{claim}{Claim}
\newtheorem{corollary}{Corollary}
\newtheorem{conclusion}{Conclusion}

\makeatletter
\renewenvironment{proof}[1][\proofname]{\par
    \pushQED{\qed}%
    \normalfont \topsep6\p@\@plus6\p@\relax
    \trivlist
    \item\relax
    {\bfseries\boldmath
        #1\@addpunct{.}}\hspace\labelsep\ignorespaces
}{%
    \popQED\endtrivlist\@endpefalse
}
\makeatother

\newcommand{\clique}{\textsf{Congested Clique}\xspace}
\newcommand{\congest}{\textsf{Congest}\xspace}
\newcommand{\CC}{\clique}
\newcommand{\QCC}{\textsf{Quantum Congested Clique}\xspace}

\newcommand{\CRR}{\mathcal{R}}
\newcommand{\AC}{\mathcal{A}}
\newcommand{\UU}{\mathcal{U}}
\newcommand{\QQ}{\mathcal{Q}}

\newcommand{\nr}{n^{\rho}}
\newcommand{\bz}{\beta_0}

\newcommand{\eqdef}{\overset{\mathrm{def}}{=}}

\newcommand{\MM}[1]{\mathsf{MM}\brak{#1}}

\newcommand{\RM}[1]{\mathsf{RM}\brak{#1}}

\newcommand{\brak}[1]{\left(#1\right)}

\newcommand{\Exp}[1]{\mathbb{E}\left[ #1 \right]}

\renewcommand{\Pr}[1]{{\mathrm{Pr}}\left[ #1 \right]}
\newcommand{\set}[1]{\left\{ #1 \right\}}

\newcommand{\BO}[1]{\mathcal{O}\brak{#1}}
\newcommand{\BOs}[1]{\mathcal{O}(#1)}
\newcommand{\TO}[1]{\tilde{\mathcal{O}}\brak{#1}}
\newcommand{\TOs}[1]{\tilde{\mathcal{O}}(#1)}
\newcommand{\Omc}[1][1]{\Omega\brak{#1}}
\newcommand{\zrn}[1]{\set{0,1,\ldots,#1}}

\renewcommand{\r}[1]{\mathcal{R}\brak{#1}}

\newcommand{\Hp}{\mathcal{H}_p}

\newcommand{\FH}{\mathcal{D}}
\newcommand{\As}{\mathtt{A}}
\newcommand{\Bs}{\mathtt{B}}
\newcommand{\Vx}{V_{\mathcal{C}_h}}

\newcommand{\omegaval}{2.371552}
\newcommand{\rhoval}{0.15667} 
\newcommand{\alphaval}{0.321334}
\newcommand{\betaval}{0.660667} %

\newcommand{\aConst}{0.4617} %
\newcommand{\bConst}{0.1567} %
\newcommand{\dt}{1.82408}
\newcommand{\tConst}{\aConst}

\newcommand{\eps}{\varepsilon}

\DeclareMathSymbol{\mhyphen}{\mathord}{AMSa}{"39}

\newcommand{\algA}{\textsf{Find-Cycle}\xspace}
\newcommand{\malgA}[1]{\ensuremath{\mathsf{Find\mhyphen Cycle}\brak{#1}}\xspace}
\newcommand{\algC}{\textsf{Find-Vertex-In-Cycle}\xspace}

\newcommand{\SP}{\mathsf{Single\mhyphen Product}}

\makeatletter
\newcommand{\whp}{%
  w.h.p.\@ifnextchar.{\@gobble}{\xspace}%
}
\makeatother

\newcommand{\true}{\textnormal{``\textsf{True}''}\xspace}
\newcommand{\false}{\textnormal{``\textsf{False}''}\xspace}

\DeclarePairedDelimiter\abs{\lvert}{\rvert}%

\DeclarePairedDelimiter{\ceil}{\lceil}{\rceil}

\makeatletter
\let\oldabs\abs
\def\abs{\@ifstar{\oldabs}{\oldabs*}}

\crefname{ineq}{inequality}{inequalities} %
\crefname{proof}{Proof}{Proofs} %
\crefname{line}{Line}{Lines}
\crefname{claim}{Claim}{Claims}
\crefname{figure}{Figure}{Figures}

\NewDocumentEnvironment{mytheorem}{m}%
  {\renewcommand{\thetheorem}{$\mathbf{#1}$}%
   \begin{theorem}
  }%
  {\end{theorem}}
\author{Keren Censor-Hillel \thanks{Department of Computer Science, Technion. \texttt{ckeren@cs.technion.ac.il}. The research is supported in part by the Israel Science Foundation (grant 529/23).} 
\and
Tomer Even \thanks{Department of Computer Science, Technion. \texttt{tomer.even@campus.technion.ac.il}.} 
\and 
Virginia Vassilevska Williams \thanks{Massachusetts Institute of Technology, Cambridge, MA, USA. \texttt{virgi@mit.edu}. Supported by NSF Grant CCF-2330048, BSF Grant 2020356, and a Simons Investigator Award.}}
\begin{document}

\title{Faster Cycle Detection in the Congested Clique}
\date{}

\maketitle
\begin{abstract}
    We provide a fast distributed algorithm for detecting $h$-cycles in the \textsf{Congested Clique} model, whose running time decreases as the number of $h$-cycles in the graph increases. In undirected graphs, constant-round algorithms are known for cycles of even length. Our algorithm greatly improves upon the state of the art for odd values of $h$. Moreover, our running time applies also to directed graphs, in which case the improvement is for all values of $h$. Further, our techniques allow us to obtain a triangle detection algorithm in the quantum variant of this model, which is faster than prior work.
    
A key technical contribution we develop to obtain our fast cycle detection algorithm is a new algorithm for computing the product of many pairs of small matrices in parallel, which may be of independent interest.
\end{abstract}

\setcounter{tocdepth}{3}
\tableofcontents

\section{Introduction}
Finding small subgraph patterns is a fundamental computational task, with a multitude of applications for uncovering connections between elements in a data set. Research has been thriving, addressing the complexity of different variants of subgraph isomorphism for fixed size subgraph patterns $H$ in a larger host graph $G$: detecting whether a copy of $H$ exists, listing all of its copies, counting the number of occurrences, and more.

\sloppy{In this paper, we provide a fast distributed algorithm for detecting $h$-cycles
    in the \clique model \cite{lotker2006minimum}, in which
    $n$ machines communicate by sending $\BOs{\log{n}}$-bit messages to each other, in synchronous rounds. }

The pioneering work of \cite{dolev2012tri} showed that all copies of any
fixed $h$-vertex graph $H$ in an $n$ node graph can be listed in this model within $\BOs{n^{1-2/h}}$ rounds. This result of course applies also to the detection variant. For the case when $H$ is a cycle, \cite{Censor-HillelKK19} provided an $h$-cycle detection algorithm running in $2^{\BOs{h}}n^{\rho}$ rounds, for both undirected and directed graphs (henceforth digraphs). Here, $\rho$ is the exponent of distributed fast matrix multiplication (FMM) in the \clique model, i.e., the value such that $\BOs{n^{\rho}}$ rounds are sufficient for multiplying two $n\times n$ matrices. The value of $\rho$ is currently known to be at most $1-2/\omega$ where $\omega$ is the centralized fast matrix multiplication exponent, and since $\omega\leq \omegaval$ \cite{alman2024asymmetry}, we get a bound for $\rho$ of $\rhoval$. 
In the case of 4-cycles in undirected graphs, \cite{Censor-HillelKK19} obtained a constant-round detection algorithm, and this result was later generalized
by \cite{Censor-HillelFG20} to hold for detection of any even-length cycle in undirected graphs.

This leaves the complexity of odd-cycle detection as an open question, as well as the detection of cycles of any length in digraphs. For triangles, \cite{dolev2012tri} showed a detection algorithm that completes within $\TOs{n^{1/3}/(t^{2/3}+1)}$ rounds, \whp\footnote{High probability in this paper refers to a probability that is at least $1-1/n^c$ for some constant $c\geq 1$.}, where $t$ is the number of triangles. Since \cite{DruckerKO13} hints that lower bounds for $H$ detection in this model are not within reach, it remains open whether the above is optimal.
\begin{center}
    {\emph{\textbf{Question}: For a given graph $H$, is there a faster $H$-detection algorithm when the number of instances of $H$ in the input graph is large?}}
\end{center}

~\\We answer this question in the affirmative, providing a fast $h$-cycle detection algorithm whose complexity decreases as the number $t$ of instances of $H$ grows. Our algorithm has the same running time for detecting $h$-cycles in graphs as well as in digraphs.
For triangles, the complexity of our algorithm greatly improves upon that of \cite{dolev2012tri}. For larger odd cycles in graphs, as well as cycles of any length in digraphs, to the best of our knowledge, this is the first improvement over \cite{Censor-HillelKK19}.

An important insight of our main technical contribution is to identify
a new refined parameter as a key player for detection: the number of vertices $x$ that participate in an $h$-cycle. Below, we elaborate on our result and technical approach.%

\subsection{Our Contributions and Technical Overview.}
\begin{figure}[t]
    \center
    \begin{tikzpicture}
        \begin{axis}[
                axis lines=middle,
                xlabel={$\log_n(t)$},
                ylabel={$\log_n(\#\mathsf{rounds})$},
                xlabel style={at={(ticklabel* cs:1.2)}, anchor=east},
                ylabel style={at={(ticklabel* cs:1.15)}, anchor=north},
                xmin=0, xmax=0.51,
                ymin=-0.0, ymax=0.2, %
                domain=0:10,
                samples=100,
                legend pos=outer north east,
                grid=both, %
                grid style={line width=.1pt, draw=gray!10}, %
                major grid style={line width=.2pt,draw=gray!50}, %
                minor tick num=5, %
                yticklabel style={/pgf/number format/fixed}
            ]

            \def\offsetlen{.8cm}
            \def\offsetsep{.15cm}
            \def\overshoot{.6cm}
            \def\ALPHA{0.3139}
            \newcommand{\middeltaval}{\dt}
            \def\colA{teal}
            \def\colC{violet}

            \addlegendimage{black, thick}
            \addlegendentry{$\TO{(n/t^2+1)^{1/3}}$ \cite{dolev2012tri}};
            \addlegendimage{red, thick}
            \addlegendentry{$\BO{\nr}$ \cite{Censor-HillelKK19}};

            \addlegendimage{\colA, thick, mark=square}
            \addlegendentry{Approach One + \Cref{thm2:mm runtime}};
            \addlegendimage{\colC, thick, mark=triangle}
            \addlegendentry{Approach Two};
            \addlegendimage{blue, thick, mark=o}
            \addlegendentry{\Cref{thm:main in t}};

            \def\R{\rhoval} %
            \def\B{\betaval} %

            \addplot[name path=f2, \colA, thick, mark=square] {(\R)-x*(2+2*(\R-1))/(3-2)} ;
            \addplot[name path=g0, \colC, thick, mark=triangle] {(\R)-x*\R/((1-\B)*(3-0))}  ;

            \addplot[name path=DLP, black, thick] {1/3-2*x/3};
            \addplot[name path=MM, red, thick] {\R};
            \addplot[name path=meet, blue, thick, mark=o] {(\R)-x*\R/((1-\B)*(3-\middeltaval))} ;
        \end{axis}
    \end{tikzpicture}
    \caption{ An illustrative comparison between our results and prior work, for the case of triangles. For each algorithm, we plot the base-$n$ logarithm of the number of rounds as a function of the base-$n$ logarithm of the number of triangles.}
    \label[figure]{fig1:comp}
\end{figure}
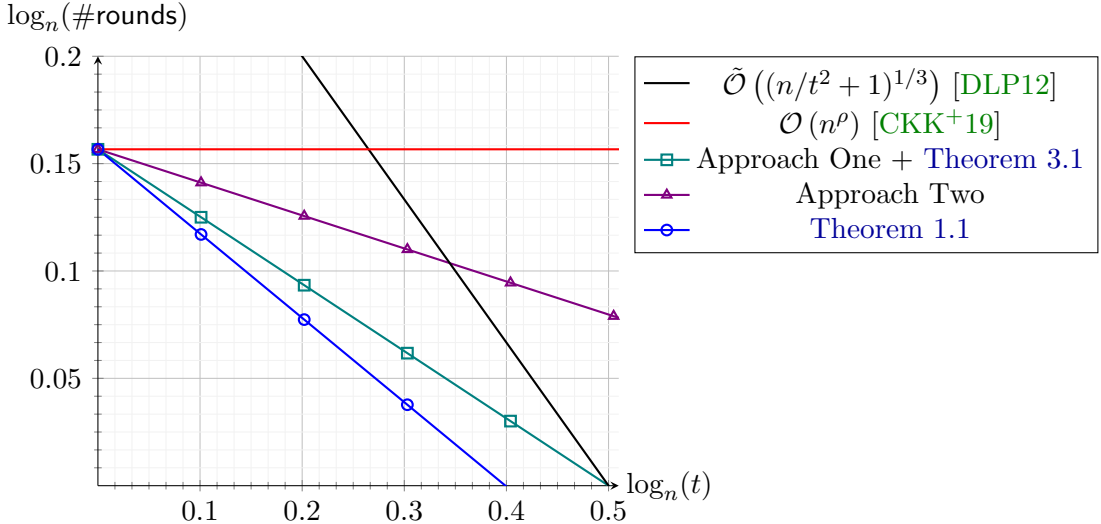

To frame our technical contributions, we first briefly overview the two previous approaches for the case of \emph{triangles}. In \cite{dolev2012tri}, an $\TO{n^{1/3}/(t^{2/3}+1)}$-round algorithm is presented, where $t$ is the number of triangles. The algorithm samples $n$ induced subgraphs, and for each sample it checks for a triangle by letting a dedicated vertex collect the edges of the sample.
In \cite{Censor-HillelKK19}, a $2^{\BO{h}}n^{\rho}$-round algorithm is presented which employs fast matrix multiplication over the entire graph.

~\\\textbf{Warm-up.} As a warm-up, consider the following combination of these approaches to get the best of both worlds,  leading to an algorithm that completes in $\TO{n^{\rho}/(t^{2\rho}+1)}$ rounds, which is already an improvement over the prior state of the art (recall that $\rho < 1/3$).
To obtain this, we sample only $t^2$ induced subgraphs to which each vertex is added with probability $1/t$.
Using the second moment method, we can show that at least one sampled subgraph contains a triangle with probability $\Omc$.
To check if each subgraph contains a triangle we use matrix multiplication, and so no vertex has to collect the edges of an entire sample.
To compute the product of $t^2$ square matrices of size $n/t$, we develop a new algorithm, which computes the product of $s$ pairs of square matrices of size $k$ in
$\BO{n^{\rho-2}\cdot k^2 \cdot s^{1-\rho}}$ rounds. This algorithm may be of independent interest.

Another natural approach to consider is one that samples a subset of vertices and checks whether any of these vertices participates in a triangle, rather than attempting to sample a complete triangle.
Here, we can sample each vertex with probability $1/t^{1/3}$ and check whether it is in a triangle by invoking rectangular matrix multiplication.
While this too improves upon the state of the art for some graphs, it is always slower than our first approach.
Note that trying to reduce the running time by sampling with a probability that is smaller than $1/t^{1/3}$ would reach a dead-end since it is not likely to hit any vertex in a triangle in case all $t$ triangles are induced by a clique of $t^{1/3}$ vertices. See \Cref{fig1:comp} for a comparison of the two approaches, as well as the previous algorithms, and our new algorithm.

A caveat is that the first approach does not extend for $h$-cycles, as the number of samples it needs to perform increases with $h$. For example, for $C_4$ detection, if we sample uniformly random induced subgraph with $n/t$ vertices, it contains a copy of $C_4$ with probability at least roughly $1/t^3$ which means we have to sample at least $t^3$ subgraphs to ensure that we find a copy of $C_4$ with a constant probability. This is slower than the previous best known algorithm of \cite{Censor-HillelKK19} which takes $2^{\BO{h}}n^{\rho}$ rounds. In other words, the first approach is slower as $h$ is larger.

~\\\textbf{Our contribution.} Our key insight is that we can further boost these two approaches such that they complement each other, in the following sense. For a fixed value of $t$, the first approach is better when the number of vertices that participate in a triangle, which we denote by $x$, is small, while the second approach is better when $x$ is large.
This refinement of considering the parameter $x$ along with $t$ allows us to bring these two approaches a big leap forward by obtaining a faster algorithm for triangle detection, as well as an algorithm for $h$-cycle detection for longer cycles, in both graphs and digraphs.

Our first algorithm, which we refer to as \algA, follows the first approach.
It samples $s=x^3/t$ subsets of vertices $(U_1,\ldots, U_s)$, by adding each vertex to $U_i$ independently with probability $1/x$.
The algorithm then checks for every $i\in[s]$ if $G[U_i]$ contains a triangle.
This involves computing the product of $s$ pairs of square matrices of size $n/x$ each, which we do in 
$
\TOs{n^{\rho-2}\cdot (n/x)^2 \cdot s^{1-\rho}}
=\TOs{n^{\rho}\cdot (1/x)^2 \cdot (x^3/t)^{1-\rho}}$
rounds.
Using the second moment method (which is very similar to Chebyshev's inequality) we show that at least one of the $s$ induced subgraphs contains a triangle with a constant probability.

Our second algorithm, which we refer to as \algC, follows the second approach. It samples a subset of vertices $S$, by adding each vertex to $S$ independently with probability $1/x$. The algorithm then checks if one of the vertices from $S$ participates in a triangle by computing the product of a rectangular matrix of size $n/x\times n$ and a square matrix of size $n$. Interestingly, the algorithm \algC also achieves the same round complexity, as a function of $x$, for $h$-cycle detection, for $h=\BO{1}$.

Our final algorithm alternates between the two algorithms until one of them terminates.
Among all $n$ vertices with $t$ triangles, the final algorithm is the slowest when the two algorithms have the same round complexity, which happens when $x^{3-\dt}=\Theta(t)$.

The following states the running time of our fast $h$-cycle detection algorithm, and is proven in \Cref{sec:h cycle}.
\begin{restatable}[$h$-Cycle Detection]{theorem}{ThmMainT}\label{thm:main in t}
    \sloppy{Let $G$ be a (directed) graph with $t$ copies of $h$-cycles.
        There is a randomized \CC algorithm for $h$-cycle detection, which takes
        $\TO{h^{\BO{h}} \cdot n^{\bConst}/ (t^{\frac{\tConst}{h-\dt}} + 1)}$
        rounds \whp.}
\end{restatable}

Here, the constants $\bConst,\aConst$ arise from the complexity of rectangular multiplication.
We are able to show that the product of a rectangular matrix of size $k\times n$ and a square matrix can be computed in $O(n^{\bConst}/k^{\aConst})$ rounds, using the formula of \cite{le2016further} and an adaptation of the code of \cite{Complexity}. %
To get a flavor of the above complexity, note that a crucial implication of \Cref{thm:main in t} is that we detect a triangle in $\TO{1}$ rounds for graphs with at least $t=\Omega(n^{0.3992})$ triangles,
improving upon the previously known threshold of $t=\Omega(n^{1/2})$ from \cite{dolev2012tri}.

~\\\textbf{Many Matrix Multiplications in Parallel.}
To implement our \algA algorithm, we need to compute the product of many small random square matrices, which are submatrices of the adjacency matrix of the input graph. We state this informally in the following theorem.
\begin{theorem}[Informal]\label{thm:informal}
    Let $k,s$ be two integers such that $k\in[\sqrt{n},n]$, and $s\leq (n/k)^2$.
    Then, in the \CC model, the $n$ vertices can compute the product of $s$ pairs of square matrices of size $k$ in $\BO{n^{\rho-2}\cdot k^2 \cdot s^{1-\rho}}$ rounds, given that the input is distributed among the vertices in a ``balanced'' manner.
\end{theorem}

The formal definitions and the proof of \Cref{thm:informal} appear in \Cref{sec:FMM}, as well as the definitions and claims we need for the proof of \Cref{thm:main in t}.

The conceptual contribution of \Cref{thm:informal} is as follows.
On one hand, it is known that $n$ vertices can compute the product of $n$ pairs of matrices of size $\sqrt{n}$ in a constant number of rounds, given that the input is balanced, by letting the $i$-th vertex collect all the entries of the $i$-th pair of matrices and computing their product.
On the other hand, $n$ vertices can also compute the product of one square matrix of size $n$ in $\BO{\nr}$ rounds, as shown in \cite{Censor-HillelKK19}.
\Cref{thm:informal} gives a smooth trade-off between these two extremes.

Note that \cite{le2016further} provides an algorithm for computing the product of $s$ pairs of square matrices of size $n$ in $\BO{n^{\rho}\cdot s^{1-\rho}}$ rounds for $s\leq n$. The paper also provides an algorithm for computing the product of $s$ pairs of \emph{rectangular} matrices of sizes $n\times m$ and $m\times n$, where the bound on the round complexity is more involved and is not given by an analytic expression, see \Cref{ssec:rm} for a discussion.
Moreover, we need to compute the product of square matrices of size smaller than $n$, which is not covered by the above algorithm.

Before providing intuition about our proof of \Cref{thm:informal}, we explain what we mean by a balanced input and how the output should be distributed. Given a set of $s$ pairs of square matrices $\QQ=\set{(S_i,T_i)}_{i\in [s]}$ of size $k$ each, where $sk^2\leq n^2$, we think of the input as a ``flat'' array of $sk^2$ entries. The input is distributed as follows. The input given to the first vertex is the first $n$ entries in this array. The second vertex gets the next $n$ entries and so on.
For the output $\set{(P_i)}_{i\in [s]}$, where $P_i=S_i\cdot T_i$ for $i\in[s]$, we again transform the set of output matrices into a flat array and let each vertex learn distinct consecutive $n$ entries from it.
Note that as long as each vertex holds unique $n$ entries from the input, and every vertex knows which entries from the input every other vertex holds, then the input can be redistributed in $\BO{1}$ rounds, using \Cref{lemma:routing}.
We therefore define a balanced input as such.
\begin{definition}[Balanced Input]\label{def:balance-informal}
    An input for $n$ vertices is \emph{balanced} if it is partitioned between the vertices such that each vertex holds at most $n$ (unique) entries from the input, and every vertex knows which entries from the input are held by every other vertex.
\end{definition}
Now we can give the intuition behind the proof of  \Cref{thm:informal}. We partition the $n$ vertices into $s$ sets of size $n/s$ each. For every $i\in[s]$ we call the $i$-th set in the partition the $i$-th \emph{team}. The $i$-th team is responsible for computing the product $(S_i,T_i)$ (this  partitioning method is similar to that of \cite{le2016further}).
After partitioning, the problem boils down to computing one product of square matrices of size $k$ using $n/s$ vertices with bandwidth of size $s\log n$, which we solve by extending the work of \cite{Censor-HillelKK19}, which considers only the product of square matrices of size equal to the number of vertices.

Recall that our main motivation for this tool of \Cref{thm:informal} is to implement the algorithm \algA.
That is, given a graph $G$ with $n$ vertices, we sample $s$ subsets of vertices $(U_1,\ldots, U_s)$, where each vertex joins each set independently with probability $p$. Each set $U_i$ defines an induced subgraph $G[U_i]$ with an adjacency matrix $A_i$.
We need to compute $(A_i)^h$ for every $i\in[s]$, and we denote the set $\mathcal{Q}=\set{(A_i,A_i)}_{i\in[s]}$ as the input, where we assume $s\leq 1/p^2$.
In order to implement the algorithm \algA using \Cref{thm:informal}, we need to show that $\mathcal{Q}$ is a balanced input. However, this does not precisely hold, but we can show a sufficient guarantee of $\mathcal{Q}$ being ``almost" balanced \whp, in which the requirements in \Cref{def:balance-informal} are weakened. These weaker conditions still allow us to quickly redistributed the elements into a balanced input in $\BO{\log n}$ rounds \whp.

~\\To conclude, we obtain a fast algorithm for $h$-cycle detection, which beats the previous state-of-the art for odd-length cycles, as well as directed cycles. The algorithm is faster as the number of copies of $h$-cycles increases, where the key parameter for the algorithm and the analysis is the number of vertices in the graph that participate in an $h$-cycle.

~\\\textbf{Triangle Detection in the \QCC Model.}
Our new matrix multiplication tool turns out to be helpful for additional tasks. In the quantum setting, we obtain the following in the \QCC model, which is similar to the \CC model, but the vertices exchange messages of $\BO{\log n}$ qubits in each round instead of standard bits.
\begin{restatable}{theorem}{ThmQInCycle}\label{thm:q2}
    There exists a \QCC $\TO{(n/(t^2 + 1))^{3\rho/4}}$-round algorithm for triangle detection, with a success probability of at least $1/2$.
\end{restatable}
We prove this theorem in \Cref{sec:q}.
\Cref{thm:q2} obtains an $\TOs{n^{3\rho/4}}$-round algorithm that remains effective even when $t=0$. This is faster than the previous state-of-the-art algorithm from \cite{Censor-HillelKK19}, which takes $\BO{\nr}$ rounds and does not leverage the additional capabilities that the \QCC model offers.
For subgraphs other than triangles, there are detection algorithms in the \QCC which use the extra power of the model. For example, \cite{Censor-HillelFG22} provides an algorithm for larger clique detection.
Moreover, in the quantum \congest model, there are algorithms for clique and cycle detection \cite{izumi2017triangle,fraigniaud2024even}, which are faster than their \congest counterparts.

Our algorithm uses a Grover search \cite{grover1996fast}, which is a quantum algorithm to find an element in an unsorted list of size $L$, while accessing only $\sqrt{\abs{L}}$ entries from the list.
Generally, given a function $f:X\to\set{0,1}$ and a universe $X$, Grover search is a quantum algorithm which finds an $x\in X$ such that $f(x)=1$ (assuming such $x$ exists), by querying $f$ at most $\sqrt{\abs{X}}$ times \whp.
In \cite{le2018sublinear} a distributed implementation of Grover search was provided, for the quantum \congest model, which was later extended to the \QCC model in \cite{IzumiG19,Censor-HillelFG22}.

We provide an overview of our algorithm, which has two steps. In the first step, we sample $s=1/t^2$ random induced subgraphs $(U_1,\ldots,U_s)$, where each vertex joins every set independently with probability $1/t$.
We partition the vertices into $s$ sets, each of $n/s$ vertices, which we call teams.
For $i\in[s]$, the $i$-th team uses Grover search to detect a triangle in the subgraph $G[U_i]$, as follows. It samples $\ell=8\log n/q^3$ subsets of vertices $(W_1,\ldots, W_\ell)$, where each vertex from $U_i$ joins each set independently with probability $1/q$. The universe for the search is the set $X\triangleq\set{G[W_i]}_{i\in[\ell]}$, and the boolean function $g$ is defined as $g(G[W_i])=1$ if $G[W_i]$ contains a triangle, and $0$ otherwise.

The correctness of the algorithm follows because if the graph $G$ has $t$ triangles, then there exists an index $i\in[s]$ such that $G[U_i]$ contains a triangle with probability at least $1/10$. The $i$-th team finds this triangle using $\sqrt{1/q^3}$ evaluations of $g$ \whp.
The final detail of the algorithm is to set $q$ such that each evaluation of $g$ takes $\BO{1}$ rounds, which optimizes the round complexity of this approach.

Our above fast triangle detection algorithm for the \QCC model is given in \Cref{sec:q}. Extending our approach to $h$-cycle detection is not straightforward; we elaborate on the challenges in \Cref{subsec:q-challenges}.

~\\\textbf{Additional Related Work.}
In this \CC model, matrix multiplication was first studied by \cite{DruckerKO13}. After the aforementioned works of \cite{Censor-HillelKK19,le2016further}, the work of \cite{censor2020sparse} showed an algorithm whose running time improves with the sparsity of the input matrices, and \cite{Censor-HillelDK21} showed algorithms for sparse matrix multiplication which also enjoy the sparsity of the output or when only a sparse piece of the output is needed.

The task of listing subgraphs has also received great attention in the \CC model. Here, each vertex needs to output a list of copies of the subgraph $H$, such that the union of the lists is exactly the set of all copies of $H$ in the graph. As mentioned, \cite{dolev2012tri} give an algorithm for listing all $h$-vertex graphs within $O(n^{1-2/h})$ rounds. For triangles, this is known to be tight by \cite{izumi2017triangle,pandurangan2021distributed}.
This optimality extends to larger cliques due to \cite{fischer2018possibilities}.
Listing triangles in sparse graphs can be done faster, as first shown by \cite{pandurangan2021distributed} with a randomized algorithm, and then followed up by \cite{censor2020sparse} with a deterministic algorithm. Afterward, \cite{censor2020distributed} showed faster sparse listing for larger cliques, which also has a deterministic algorithm due to \cite{Censor-HillelFG20}.

We mention that in the closely-related \congest model, in which the communication graph is the input graph itself, rather than being a complete network, the state of affairs is in stark contrast to the \CC model. For listing cliques, optimal algorithms are known due to \cite{chang2021near,censor2021tight}, with deterministic solutions in \cite{ChangS20,censor2022deterministic,chang2024deterministic}. %
The underlying approach, initiated by \cite{chang2021near}, is to construct an expander decomposition, which partitions the vertices into components with good expansion (low mixing time). At a very high level, the vertices of each component list cliques for which some edges are inside the component, and then recurse over the remaining edges.
However, for an algorithmic approach of the \CC model to have a fast implementation in the \congest model, also when using the known routing procedures \cite{ghaffari2017distributed,ghaffari2018new}, the algorithm has to adhere to certain conditions. In other words, it is not the case that any algorithm in the \CC model can be executed efficiently by the components of the expander decomposition in the \congest model. 

Specifically, we stress that for the detection variant in \congest, the state of the art even just for triangles is the same as for listing. That is, even the $O(n^{\rho})$ algorithm of \cite{Censor-HillelKK19} does not have an implementation in the \congest model, and it is unknown how to detect triangles in less than the time it takes for listing them (interestingly, the only lower bounds that are known are that a single round does not suffice \cite{AbboudCKL20,fischer2018possibilities}). For larger $h$-cycles, detection for odd values of $h$ is known to have a complexity of $\tilde{\Theta}(n)$ \cite{drucker2014power}. Much work is invest in studying the complexity of detecting even cycles 
\cite{drucker2014power,KorhonenR17,Censor-HillelFG20,eden2022sublinear,Censor-HillelFG22,ApeldoornV22,fraigniaud2024even}, with the state of the art being a recent result showing that $h$-cycles can be detected in $\tilde{O}(n^{1-2/h})$ rounds for even values of $h$ \cite{fraigniaud2024even}.

Investigations into the subgraph detection problem have also been conducted in additional models such as the quantum \congest and \QCC models, where vertices exchange qubits instead of standard bits. In \cite{izumi2020quantum}, a quantum \congest $\TO{n^{1/4}}$-round algorithm for triangle detection was presented, which outperforms the $\TO{n^{1/3}}$-round \congest algorithm. This approach was further improved in \cite{Censor-HillelFG22} by developing an $\TO{n^{1/5}}$ rounds quantum algorithm. 
Additional upper and lower bounds for cycle detection in the quantum \congest model were presented in \cite{ApeldoornV22,fraigniaud2024even}. 
In \cite{Censor-HillelFG22}, an \QCC algorithm for $p$-clique detection, for $p\geq 4$, was presented, which achieves an $\TO{n^{1-2/(p-1)}}$-round complexity, which is faster than the classical \CC algorithm.
Further research in the quantum distributed models includes both upper and lower bounds for various problems \cite{ElkinKNP14,IzumiG19,ApeldoornV22}.

There is extensive research about subgraph finding in additional models of distributed computing. All of these important works are a bit more far from our work here, and hence we refer the reader to the survey of \cite{C21}, which contains a recent overview of subgraph finding algorithms for distributed settings.

\section{Preliminaries}
\label{subsec:preliminaries}

We use $[n]$ to denote the set $\set{1,\dots, n}$.
We denote the base graph by $G$, its vertices by $V(G)$, where unless stated otherwise we assume $V(G)=[n]$.
Fix some constant integer $h\geq 0$.
Let $t$ denote the number of $h$-cycles in $G$, and let
$\Vx(G)$ denote the set of vertices that participate in an $h$-cycle in $G$, where we denote by $x$ the size of the set $\Vx(G)$. This parameter plays a crucial role in the analysis in \Cref{sec:h cycle}.
We establish a connection between the two parameters $x$ and $t$ as follows.
\begin{definition}[$\delta$]\label{def:x-delta}
    For undirected graph, $G$ with $t$ copies of $h$-cycle, and $\Vx(G)=x$,
    we define $\delta$ as such that $x^{h-\delta}=2h\cdot t$.
    If $G$ is directed, we define $\delta$ as such that $x^{h-\delta}=h\cdot t$.
\end{definition}
We present two claims on $\delta$. We show that $\delta\in[0,h-1]$, and that the term $x^{-\delta}$ is equal to is the probability for $h$ vertices sampled uniformly at random with replacement from $\Vx(G)$ to form an $h$-cycle in $G$.

\begin{claim}\label{claim:delta prop0}
    Let $G$ be a graph with $t$ copies of an $h$-cycle and $x$ vertices that participate in at least one $h$-cycle. 
    Sample $h$ vertices $(v_1,\ldots, v_h)$ from $\Vx(G)$ uniformly at random with replacement from $\Vx(G)$. Then the probability that they form an $h$-cycle is exactly $x^{-\delta}$. 
\end{claim}
\begin{proof}
    We first emphasize that we consider $(v_1,\ldots, v_h,v_1)$ to be an $h$-cycle only if all the vertices are distinct other than the first and the last which are the same vertex, and that for $i\in[h]$ we have that $(v_i,v_{i+1 \mod h})$ is an edge in $G$.

    To prove the claim, we count the number of $h$-tuples of vertices that form an $h$-cycle in $G$, and the number of $h$-tuples of vertices from $\Vx(G)$. The probability that the sampled vertices form an $h$-cycle is equal to the ratio of the former to the latter.
    
    The latter is equal to $x^h$. The former is equal to $2ht$ for undirected graphs, as each $h$-cycle is associated with exactly $2h$ tuples of vertices. 
    For directed graphs, the former is equal to $ht$, as each $h$-cycle is associated with exactly $h$ tuples of vertices. 
    
    For undirected graphs, this implies that the sampled vertices form an $h$-cycle is equal to $\frac{2ht}{x^h} =\frac{x^{h-\delta}}{x^h}=x^{-\delta}$, where we used the fact that $2ht=x^{h-\delta }$.
    For directed graphs, this implies that the sampled vertices form an $h$-cycle is equal to $\frac{ht}{x^h} =\frac{x^{h-\delta}}{x^h}=x^{-\delta}$, where we used the fact that $ht=x^{h-\delta }$.
    This completes the proof.
\end{proof}

\begin{claim}\label{claim:delta prop2}
    It holds that $\delta\in [0,h-1]$.
\end{claim}
\begin{proof}
    It is clear that $\delta\geq 0$, as $x^{-\delta}$ is equal to the probability of some event and therefore cannot be strictly bigger than $1$.
    We prove that $\delta\leq h-1$.
    Sample $h$ vertices from $\Vx(G)$ with replacement and let $u$ be the first vertex that is sampled. Because $u$ participates in at least one $h$-cycle, there is an ordered tuple of $h-1$ vertices $(v_2,v_3,\ldots, v_h)$ such that $(u,v_2,v_3,\ldots, v_h,u)$ forms an $h$-cycle. The probability that the sampled vertex is $v_i$ for every $i \in {2, 3, \ldots, h}$ is exactly $1/x^{h-1}$. This lower bounds the probability of sampling an $h$-cycle. Therefore, $x^{-\delta}\geq 1/x^{h-1}$, implying $\delta\leq h-1$.
\end{proof}%

\subsection{Additional Tools}
\begin{lemma}[Lenzen's Routing Lemma \cite{lenzen2013optimal}]\label[lemma]{lemma:routing}
    The following is equivalent to the \clique model:
    In every round, each vertex can send (receive) $\set{b_i}_{i\in[n]}$ bits to (from) the $i$-th vertex, for any sequence $\set{b_i}_{i\in[n]}$ satisfying $\sum_{i=1}^n b_i=\BO{n\log n}$. In other words, any routing scheme in which no vertex sends or receives more than $\BO{n}$ messages can be preformed in $\BO{1}$ rounds.
\end{lemma}

\begin{theorem}[Chernoff Bound {\cite[Corollary 1.10.6.]{doerr2019theory}}]\label{thm:chernoff}
    Let $X_1, \ldots, X_n$ be independent random variables taking values in $[0,1]$ and $X=\sum_i X_i$. Let $\delta\in[0,1]$. Then,
    $\Pr{\abs{X-\Exp{X}\geq \delta\Exp{X}}}\leq 2\exp\brak{-\delta^2\cdot \Exp{X}/3}$
\end{theorem}

\begin{theorem}[Reverse Markov's inequality {\cite[(1.6.4)]{doerr2019theory}}]\label{thm:rev mark}
    Let $X$ be a random variable with support contained in $[0,M]$.
    Then, for $R\in\mathbb{R}$, we have $\Pr{X>R}\geq \frac{\Exp{X}-R}{M-R}$.
\end{theorem}

\begin{theorem}[FMM-based triangle detection \cite{Censor-HillelKK19}]\label{thm:detect deterministic}
    There is a deterministic algorithm for triangle detection, which takes $\BO{\nr}$ rounds.
\end{theorem}

\section{Fast Matrix Multiplication in \clique}\label{sec:FMM}
\subsection{Preliminaries and Balanced Products.}\label{ssec:MM sk}
In this section, we define the problem of computing $s$ pairs of square matrices of size $k$ each, as well as defining what is a balance input.
We assume throughout the paper that the matrices are over a field $\mathbb{F}$, where each element can be represented using $\BO{\log n}$ bits.

We introduce the following definitions to specify the required  input.
\begin{definition}[The notations {$A[i,*]$ and $A[*x*,*y*]$}]
    Let $A$ be some matrix. We denote the $i$-th row of $A$ by $A[i,*]$.
    For a matrix $A$ of dimension $n\times n$ and two indices $x,y\in [\sqrt{k}]$ for $k\leq n$, we also use $A[*x*,*y*]$ to denote a matrix of dimension $n/\sqrt{k}\times n/\sqrt{k}$, which is the following submatrix of $A$. For every index $v\in [n]$, we split it into three indices $v=v_1 v_2 v_3$ where $v_1,v_3\in[n^{1/2}\cdot k^{-1/4}],\; v_2\in [\sqrt{k}]$.
    The expression $*x*$ then refers to all $v$ for which $v_2=x$.
\end{definition}

The following definition formally defines the problem of multiple matrix multiplications.
\begin{definition}[$\mathsf{Product}\brak{\QQ}$]\label{def2:product}
    Given set of $s$ pairs of square matrices $\QQ=\set{(S_i,T_i)}_{i\in [s]}$ of size $k\times k$.
    In the $\mathsf{Product}\brak{\QQ}$ problem, $n$ nodes need to compute the products of those pairs of matrices. The input is distributed as follows.
    Each vertex $v\in V$ is assigned a label $\ell(v)=(i,x,y)$, where $i\in[s]$, and $x,y\in[\sqrt{n/s}]$.
    The vertex $v$ gets as input the submatrix $S_i[*x*,*y*],T_i[*x*,*y*]$, and has to learn the entries of the submatrix $P_i[*x*,*y*]$, where $P_i=S_i\cdot T_i$ for $i\in[s]$.
    We denote the round complexity of this problem by $\MM{k,k,k;s}$.
\end{definition}
Note that every vertex can learn the label of each other vertex in $\BO{1}$ rounds.
The following theorem is the main theorem for this section, in which we provide an upper bound for the round complexity of the $\mathsf{Product}\brak{\QQ}$ problem.
\begin{theorem}\label{thm2:mm runtime}
    For any two integers $k,s$ where $k\in[\sqrt{n},n]$ and $s\leq (n/k)^2$, we have that
    \center$\MM{k,k,k;s}= \BOs{n^{\rho-2}\cdot k^2 \cdot s^{1-\rho}}$.
\end{theorem}

To prove the theorem, we first partition the $n$ nodes into $s$ sets of size $n/s$ each. For every $i\in[s]$ we call the $i$-th set in the partition the $i$-th \emph{team}. The $i$-th team is responsible for computing the $i$-th product in $\QQ$, i.e., the product $(S_i,T_i)$.
After partitioning into teams, the problem boils down to computing one product of square matrices of size $k$ using $n/s$ node with bandwidth of size $s\log n$.
This extends \cite{Censor-HillelKK19}, in which only the product of square matrices of size equal to the number of vertices is considered, and uses \cite{le2016further}, in which multiple products are divided into teams.
A crucial step in the algorithm for \Cref{thm2:mm runtime} is to compute the product of a single matrix of size $R$ using $n'$ nodes, which we define next.

\begin{restatable}[$\SP\brak{n',R}$]{definition}{DefProd}
    \label{def:problem}
    Let $H$ be a team with $n'$ vertices.
    Let $S,T$ be two square matrices of dimension $R$ for some $R\in[1,(n')^2]$, and define $P=ST$.
    Each vertex $v\in H$ is assigned a label $\ell^\prime(v)=xy$, where $x,y\in[\sqrt{n'}]$.
    The input of each vertex $v\in H$ with label $\ell^\prime(v)=xy$ is $S[*x*,*y*]$ and $T[*x*,*y*]$, and its output should be $P[*x*,*y*]$.
    We denote this problem by $\mathsf{Product}\brak{n',R}$.
\end{restatable}

\begin{restatable}[]{proposition}{SingleFMM}\label{PROP:PART2}
    The $\SP\brak{n',R}$ problem can be solved in the \clique model with $n'$ vertices in the base graph and bandwidth $B$, in $\BO{(n')^{\rho}\cdot (R/n')^2\cdot \frac{F}{B}}$ rounds, where each entry in $R$ can be represented using $\BO{F}$ bits.
    Using bandwidth $B$ means that in each round, each vertex in the base graph can send $B$ bits to every other vertex.
\end{restatable}
The proof for \Cref{PROP:PART2} (\Cref{appendix:FMM}) is derived by minor modifications to the proof of {\cite[Theorem 1]{Censor-HillelKK19}} of multiplying $n\times n$ matrices.
We are now ready to prove \Cref{thm2:mm runtime}.
\begin{proof}[Proof of \Cref{thm2:mm runtime}]
    We use \Cref{PROP:PART2} by setting $n'$ to be the size of each team, i.e., $n'=n/s$. We set $R$ to $k$. Since the base matrix is a boolean matrix, we can set $F=1$.
    We get that after step one, each team can compute the square of the matrix it is responsible for in $\BO{(n')^{\rho-2}\cdot  R^2/B}$ rounds.

    We next explain how to further improve the running time of the second step.
    In the algorithm presented in \Cref{PROP:PART2}, communication between vertices is limited to team members.
    Consequently, each vertex transmits only $\BO{n'}$ messages per round of size $\log n'$ each, as it behaves as if it is ``unaware'' that it is part of a larger set of vertices.
    However, by utilizing the bandwidth between the teams, each vertex can send more messages or equivalently send larger messages.
    That is, we can have each vertex send $\BO{n'}$ messages, each of $\BO{(n/n')\log n}$ bits.
    From \Cref{PROP:PART2}, we get that increasing the bandwidth speeds up the computation by a $n/n'=s$ factor. In other words, we can set $B=s$.
    Specifically, each team completes its computation in
    $
        \BO{(n')^{\rho-2}\cdot  R^2/B}=
        \BO{(\frac{n}{s})^{\rho-2}\cdot  k^2/s}=
        \BO{n^{\rho-2}\cdot  k^2\cdot s^{1-\rho}}$
    rounds, as needed.
\end{proof}

\subsection{Multiple Products of Random Submatrices}
In this subsection, we explain how to use the tools we developed in the previous subsection, to detect an $h$-cycle in $s$ induced subgraphs sampled uniformly at random. Specifically, we explain how to compute the $h$-th power of the adjacency matrices of those subgraphs.

Let $\UU=(U_1,\ldots, U_s)$ be a set of subsets of vertices, where each subset is a uniformly random set. That is, each vertex joins to the set $U_i$ independently and uniformly at random, with probability $p$.
For each $i\in[s]$, we denote the the adjacency matrix of $G[U_i]$ by $A_i$, and define $\QQ=\set{(A_i,A_i)}_{i\in [s]}$.

We explain how to compute the $h$-th power of $\set{A_i}_{i\in[s]}$ in parallel by reducing this problem into the $\mathsf{Product}(\QQ)$ problem.
In other words, we explain how to redistribute the initial input, into an input for the $\mathsf{Product}(\QQ)$ problem.
We show that the reduction takes $\BO{\log n}$ rounds (\Cref{prop:part1}) \whp, and provide an algorithm that tests whether the reduction algorithm can be executed in $\BO{\log n}$ or not (\Cref{claim:no-balance}). The testing algorithm takes $\BO{1}$ rounds.
In case the testing algorithm indicates that we sampled a set $\UU$ for which the reduction takes more than $\BO{\log n}$ rounds, we discard the current sample set $\UU$, and sample a new set. We also prove that \whp we will not have to discard the sampled set (\Cref{claim:balance whp}).

In \Cref{ssec:MM sk}, we described an algorithm to compute the product of $s$ pairs of matrices of size $k$ each. Here, we describe an algorithm to compute the product of $1/p^a$ pairs of matrices of size \emph{at most} $4np$  each. The connection between the parameters $s,k, p$ and $a$ is as follows.
We set $k=4np$, and $s=1/p^a$ where $a\in[0,2]$.
We get that $sk^2\leq n^2$ as desired.

We provide a definition for a set $\UU$ for which we can redistribute the input in $\BO{\log n}$ rounds. We call such a set a \emph{$p$-balanced} set.

\begin{definition}[$p$-Balanced Set]
    \label{def:p-balanced}
    Given is a parameter $p$. Let $\UU$ be a set of subsets of vertices from $V(G)$.
    Let $a$ be a constant for which $\abs{\UU}=p^{-a}$.
    We say that $\UU$ is a \emph{$p$-balanced set} if all the following conditions hold:
    \begin{enumerate}
        \item $a\in[0,2]$ (so $\abs{\UU}\leq (1/p)^2$).
        \item $n^{-1/2}\leq p\leq 1$.
        \item Every vertex $v\in V(G)$ belongs to at most $\ceil*{\abs{\UU}\cdot p}4\log n=\ceil*{p^{1-a}}4\log n$ sets in $\UU$.
        \item Every set $U\in\UU$ is of size at most $\ceil*{4np}$.
    \end{enumerate}
    Note that $(1)$ and $(2)$ imply that $\abs{\UU}\leq n$.
\end{definition}
The next claim proves that if $\UU=(U_1,\ldots, U_{p^{-a}})$ is a $p$-balanced set, then every vertex $v$ can learn the IDs of all vertices in $U_j$ for each set $U_j$ to which $v$ belongs.
\begin{claim}\label{prop2:prob-ind}
    Let $\UU=(U_1,\ldots, U_{p^{-a}})$ be a $p$-balanced set,
    where every vertex knows to which $U_j$ it belongs.
    There is an $\BO{\log{n}}$-round \clique algorithm that allows each vertex to learn the IDs of all vertices in $U_j$ for each set $U_j$ to which $v$ belongs.
\end{claim}
\begin{proof}
    Let $s=\abs{\UU}=1/p^a$.
    Recall that $V=[n]$, and define a label $L(v)\triangleq(i,j)$ for $i\in[s],j\in[n/s]$.
    We call the set of vertices with label $(i,*)$ the $i$-th team, and denote it by $T_i$. There are $s$ teams, each of size $n/s$.
    \begin{enumerate}
        \item Each vertex sends its label to all other vertices.
        \item For every $i\in[s]$, each vertex $u\in U_i$ sends an Ack to every vertex $v\in T_i$.
        \item For every $i\in[s]$, each vertex $v\in T_i$ locally partitions $U_i$ into $n/s$ parts $(S_i^1,\dots,S_i^{n/s})$, of equal size of $\abs{U_i}/(n/s)= \BO{sp}$.
              The partition is determined by the labels, and therefore can be done consistently without further communication.
        \item For every $i$, each vertex $v\in T_i$ with the label $(i,j)$ sends the set $S_i^j$ to all vertices in $U_i$.
    \end{enumerate}

    ~\\\textbf{Analysis.}
    The first step takes a single round.
    The second step takes $\BO{\log{n}}$ rounds, as follows.
    A vertex $u$ has to send $\BO{np}$ messages for every index $i$ for which $u\in U_i$.
    Since $\UU$ is $p$-balanced, there are at most $\ceil*{p^{1-a}}4\log n$ such sets.
    Therefore, each vertex has to send at most $\BO{n\log n\cdot \ceil*{p^{2-a}}}$ messages which is at most $\BO{n\log n}$ messages as $\UU$ is $p$-balanced.
    Moreover, each vertex receives at most $\BO{np}$ messages, as $\abs{U_i}= \BO{np}$.
    The third step is local, and in the last step, each vertex sends a block of $\BO{sp}$ messages to $\BO{np}$ vertices. This means no vertex sends more than $\BO{n\cdot s\cdot p^2}=\BO{n}$ messages, where we used the fact that $sp^2\leq 1$, which follows since $s= 1/p^a$, and $a\leq 2$.
    On the receiving side, each vertex receives at most $\ceil*{p^{1-a}}4\log n$ blocks of messages, where each block is of size $\BO{sp}$, which means it receives at most $\BO{s\ceil*{p^{2-a}}\log n}\leq \BO{n\log n}$ messages, as required.
    We showed that each vertex sends and receives $\BO{n\log n}$ bits. By \Cref{lemma:routing}, this routing task can be completed within $\BO{\log n}$ rounds.
\end{proof}

In what follows, we explain how to route the input of a $p$-balanced set $\UU$, after each vertex learns the IDs of all vertices in $U_j$ for each set $U_j$ to which $v$ belongs, to match the input of the $\mathsf{Product}(\QQ)$ problem. This routing takes $\BO{\log n}$ rounds.
Before providing a routing algorithm, we introduce new notation that we need in order to explain how the input is routed.

\begin{definition}[The notation {$A[i,U_j]$}]
    Recall that $V=[n]$,
    and let $U_j\subset V$, and let $i$ be some vertex in $U_j$.
    Let $A_j$ be the corresponding adjacency matrix of $U_j$.
    For vertex $i$ in the set $U_j$ we define $A[i,U_j]$ as the submatrix of $A$, which contains only the $i$-th row, and all $k$ columns, for $k\in U_j$.
\end{definition}

\begin{proposition}[Redistributing the Input]\label[proposition]{prop:part1}
    Given a parameter $p$, let $\UU\triangleq(U_1,\dots, U_{p^{-a}})$
    be a set of subsets of vertices from $V(G)$ which is a $p$-balanced set.
    For each $i \in [p^{-a}]$, let $A_i$ be the adjacency matrix of the induced graph $G[U_i]$.
    Label each vertex $v\in [n]$
    as
    $\begin{aligned}
            \ell(v)\triangleq(x,y,i)\in[\sqrt{np^a}]\times[\sqrt{np^a}]\times[p^{-a}]\;.
        \end{aligned}$
    Partition the vertices into ${p^{-a}}$ teams, each of size $n\cdot p^a$, where the $j$-th team contains all vertices $v$ with label $\ell(v)=(x,y,i)$ such that $i=j$.
    Then, in parallel, each vertex $v$ with label $\ell(v)=(x,y,i)$ can learn $A_i[*x*,*y*]$ in $\BO{\log n}$ rounds.
\end{proposition}

\begin{proof}
    Recall that because $\UU$ is a $p$-balanced set, each vertex knows its own label. Note that in $\BO{1}$ rounds, each vertex can learn the labels of all other vertices, as each label can be encoded using $\BO{\log n}$ bits.
    We show that no vertex needs to send or receive more than $\BO{n \log n}$ messages.
    By \Cref{lemma:routing}, this implies there is an $\BO{\log n}$-round algorithm for this routing problem.
    
    ~\\\textbf{Received messages.}
    Each vertex $v$ with label $(x,y,i)$ needs to learn $A_i[*x*,*y*]$. This submatrix has $\BO{np^{2-a}}$ entries:
    $A_i$ has $\ceil{4np}\cdot \ceil{4np}$ entries, and
    $x,y\in[\sqrt{np^a}]$, so $A_i[*x*,*y*]$ has $(\ceil{4np}/\sqrt{np^a})\cdot (\ceil{4np} /\sqrt{np^a})$ entries.
    In other words, in each round every vertex receives at most $\BO{np^{2-a}}= \BO{n}$ messages, as $\UU$ is $p$-balanced and therefore $a\leq 2$.

    ~\\\textbf{Sent messages.}
    Fix some vertex $u$, and a some index $i$. Let $T_i$ denote the $i$-th team.
    If $u$ is not in $U_i$, then it does not send any message to any vertex in $T_i$. Otherwise, $u$ sends
    $A[u,U_i]=A_i[u,*]$ to $T_i$, as follows.\footnote{The vertex $u$ learned the IDs of all other vertices in $U_i$, and therefore holds the entries of $A_i[u,*]$.}
    We encode $u$ as $(z_1,x,z_2)$ where $x\in [\sqrt{np^a}]$ and $z_1,z_2\in [\sqrt[4]{np^a}]$ (note that this is different from $u$'s label).
    For each vertex $v\in T_i$, with label $\ell(v)=(x,y,i)$,
    $u$ sends $A_i[u,*y*]$ to $v$.
    We emphasize that exactly one vertex from the $i$-th team receives the entry $A_i[u,z]$ for $z\in\abs{U_i}$. However, some vertices can receive multiple entries, and some get no entries from $u$.

    Let $M_u(T_i)$ denote the number of messages sent from $u$ to any vertex in $T_i$.
    Note that each entry from $u$'s input is sent to exactly one vertex from the $i$-th team, and therefore $M_u(T_i)\leq \abs{U_i}=\BO{np}$.
    Because $\UU$ is $p$-balanced, the vertex $u$ participates in at most
    $\max\set{p^{1-a},1}\cdot 4\log n$ sets from $\UU$, and for each such set $U_j$ we have $M_u(T_j)\leq \BO{np}$.
    Overall, it sends at most
    $\BO{np\cdot p^{1-a}4\log n} = \BO{np^{2-a}\log n}$
    messages, which is at most $\BO{n\log n}$ as $\UU$ is $p$-balanced and therefore $a\leq 2$.
\end{proof}

The next corollary address the detection of an $h$-cycle in one of the sampled graphs. It follows from \Cref{prop:part1,thm2:mm runtime}. The algorithmic aspects of this corollary are presented in \Cref{sec:h cycle}.
\begin{corollary}\label{thm:p balanced MM}
    Given $r$ and $p$, let $\UU=(U_1,\dots,U_r)$ be a set of subsets of vertices from $V(G)$, where $\UU$ is a $p$-balanced set, and
    every vertex in $G$ knows whether it belongs to $U_j$ for every $j\in[s]$.
    For each $i \in [p^{-a}]$, let $A_i$ be the adjacency matrix of the induced graph $G[U_i]$.
    Label each vertex $v\in [n]$ as
    $\begin{aligned}
            \ell(v)\triangleq(x,y,i)\in[\sqrt{np^a}]\times[\sqrt{np^a}]\times[p^{-a}]\;.
        \end{aligned}$
    Then, for every integer $h$, in parallel, each vertex $v$ with label $\ell(v)=(x,y,i)$ can learn $(A_i)^h[*x*,*y*]$ in $\BO{\log (h)\cdot \nr\cdot p^{2+a(\rho-1)}+\log  n}$ rounds.
\end{corollary}
\begin{proof}
    Using \Cref{prop:part1}, each vertex $v$ with label $\ell(v)=(x,y,i)$ can learn $A_i[*x*,*y*]$ in $\BO{\log n}$ rounds.
    Then, we can apply \Cref{thm2:mm runtime}, with $S_i=T_i=A_i$ for $i\in[r]$, so that each vertex $v$ with label $\ell(v)=(x,y,i)$ can learn $(A_i)^2[*x*,*y*]$ in $\BO{\nr\cdot p^{2+a(\rho-1)}}$ rounds.
    We can then reapply \Cref{thm2:mm runtime}, with $S_i=T_i=(A_i)^2$ for $i\in[r]$, so that each vertex $v$ with label $\ell(v)=(x,y,i)$ can learn $(A_i)^4[*x*,*y*]$ in $\BO{\nr\cdot p^{2+a(\rho-1)}}$ rounds.
    Using additional $\BO{\log h\cdot \nr\cdot p^{2+a(\rho-1)}}$ rounds, each vertex $v$ can learn $(A_i)^{2^\ell}[*x*,*y*]$, for $\ell\in[\log(h)]$. Afterwards, in additional $\BO{\log h\cdot \nr\cdot p^{2+a(\rho-1)}}$  rounds, each vertex can learn $(A_i)^{\ell}[*x*,*y*]$, for $\ell\in[h]$.
\end{proof}

The remainder of this subsection shows that a set $\UU$ of uniformly random subsets of vertices is a $p$-balanced set \whp.
We also provide an algorithm to test whether a set of subsets of vertices is a $p$-balanced set in a constant number of rounds.
We create $s\triangleq p^{-a}$ subsets of vertices, by letting each vertex join each set independently with probability $p$ (each vertex knows $p$ and $a$). We denote the sets by $\UU=(U_1,\dots, U_s)$, and show that $\UU$ is $p$-balanced \whp, and that the vertices can determine whether this is the case in $\BO{1}$ rounds.
\begin{claim}\label{claim:balance whp}
    $\UU$ is balanced with probability at least $1-2/n^3$.
\end{claim}
The proof of \Cref{claim:balance whp} appears in \Cref{app:proofs}.

\begin{claim}\label{claim:no-balance}
    There is a \clique algorithm that decides if $\UU$ is $p$-balanced in $\BO{1}$ rounds.
\end{claim}

\begin{proof}[Proof of \Cref{claim:no-balance}]
    Recall that $p$ and $a$ are known to all vertices, and each vertex knows to which sets in $\UU$ it belongs.
    Then a violation of the first or second condition is noticed by all vertices.
    If the third condition is violated, there exists a vertex $v$ that is in too many sets, and $v$ is aware of it.

    To decide whether the fourth condition is violated, we design the following algorithm.
    For every $i\in[s]$, each vertex $v$ in $U_i$ sends an Ack to vertex $u$ with $\mathsf{ID}(u)=i$.
    If one vertex receives more than $4np$ messages, it broadcasts the message $\bot$ to all other vertices.
\end{proof}

\subsection{Rectangular Matrices}\label{ssec:rm}
Here we set the ground for computing the product of two rectangular matrices in \clique.
We build on the work of \cite{le2016further}, which shows that computing the product of two rectangular matrices $S,T$ of size $n\times n^{\beta_0}$ and $n^{\beta_0}\times n$ respectively takes $\BO{n^{o(1)}}$ rounds.

\begin{definition}[Rectangular matrix multiplication $\RM{S,T}$]\label{def:rm st}
    Given two matrices $S,T$ of dimension $n\times n^z$ and $n^z\times n$ where $z\in[0,1]$, the $\RM{S,T}$ problem is to compute $P=S\cdot T$ in the \clique model with $n$ nodes. The input of each vertex $i\in [n]$ is $S[i,*]$ and $T[*,i]$, and its output should be $P[i,*]$.
    We abuse the notation and use it also to denote the complexity of the problem by $\MM{n,n,n^z}$ or $\RM{n^z}$.
\end{definition}
\begin{remark}
    For two matrices $S,T$ of size $n^z\times n$ and $n\times n$ (or $n\times n$ and $n\times n^z$), we get the same round complexity \cite[Theorem 6]{elkin2022centralized}.
    In this case, we assume the input is of each vertex $i\in [n]$ is $S[*,i]$ and $T[*,i]$, and its output should be $P[i,*]$.
\end{remark}

\begin{definition}[The exponent of matrix multiplication]\label{def:rect}
    The exponent of the \emph{sequential} complexity of computing the product of two matrices of dimensions $n\times n^z$ and $n^z\times n$ respectively is denoted by $\omega(z)$. We denote by $\BO{n^{\rho(z)}}$ the round complexity of computing this product in the \clique model.
    Let $\alpha_0 = \lim_{\eps\to 0}\sup\set{z\mid \omega(z)\leq 2 + \eps},$ and
    $\beta_0 = \lim_{\eps\to 0}\sup\set{z\mid \rho(z)= \eps}$. Then $\alpha_0 \geq \alphaval$ \cite{williams2023new} and $\beta_0 \geq (1+\alpha_0)/2 \geq \betaval$ \cite{le2016further,williams2023new}.
\end{definition}
We would like to upper bound the function $\rho(z)$ by some analytic function, which is easy to work with. To do so, we use the following notation.
\begin{definition}[The notation $\Bs,\;\As$]
    \label{def:line_parameters}
    We will use $\Bs,\As$ for two real non-negative constants such that, for every $y\in[0,1-\bz]$ we have $\rho(1-y)\leq \Bs-\As y$.
\end{definition}
We give two explicit linear functions which bound the function $\rho(1-y)$. First, if the function $\rho(z)$ is convex, then we can set $\As = \rho(1)/(1-\bz)$ and $\Bs=\rho(1)$, by taking the line passing through the points $(\bz,\rho(\bz))$ and $(1,\rho(1))$.
This is of course the ``best'' (minimizing $l_\infty$ norm) linear function that upper bounds the function $\rho(1-y)$. Yet, proving that $\rho(z)$ is convex is beyond the scope of this paper.
Instead, the following claim is an additional explicit linear function we provide, which does not assume that $\rho(z)$ is convex.
\begin{restatable}[]{claim}{ClaimNumericRho}\label{CLAIM:NUMERIC RHO}
    We can set $\As = \aConst$ and $\Bs=\bConst$.
\end{restatable}
The proof of \Cref{CLAIM:NUMERIC RHO} (\Cref{appendix:convexity}) is numeric: We build a step function which is always above $\rho(z)$, and then find a line which is above the step function in the desired range.
For any choice of $\Bs, \As$ that fits \Cref{def:line_parameters} and any $p\in(0,1)$ we have $
    \RM{np}
    =\BO{n^{\rho(1-\log_n(1/p)) }}
    \leq \BO{n^{\Bs - \As\cdot \log_n(1/p) }}
    = \BO{n^{\Bs }\cdot p^{\As}}$.
Thus, by the above discussion and by \Cref{CLAIM:NUMERIC RHO} we get the following.
\begin{conclusion}\label{claim:RM complexity}
    For $p\in(0,1)$ we have
    $\begin{aligned}
            \RM{np}
            \leq \BO{n^{\bConst}\cdot p^{\aConst}}\;.
        \end{aligned}$
    If $\rho(z)$ is convex, we have
    $\begin{aligned}
            \RM{np}\leq \BO{n^{\rho}\cdot p^{\rho/(1-\bz)}}.
        \end{aligned}$
\end{conclusion}

\newcommand{\Gp}{G_{\varphi}}
\renewcommand{\Hp}{F_{\varphi}}
\newcommand{\cq}{\frac{1}{h^h}}
\newcommand{\cqa}[1][1]{\frac{#1}{h^h}}

\section{$h$-Cycle Detection}\label{sec:h cycle}
In this section, we prove the following theorem.
\ThmMainT* %
We do so by presenting two algorithms and analyzing their running time and success probability as a function of the parameters $n$, $t$, and $x$. \Cref{fig2:comp} depicts their running times.
In \Cref{ssec:c1,ssec:implement algc} we prove \Cref{thm:h-cycle}, in \Cref{ssec:c3,ssec:implement alga} we prove \Cref{thm:h-cycle induced}, and in \Cref{ssec:c5} we prove \Cref{thm:main in t}, using \Cref{thm:h-cycle,thm:h-cycle induced}.
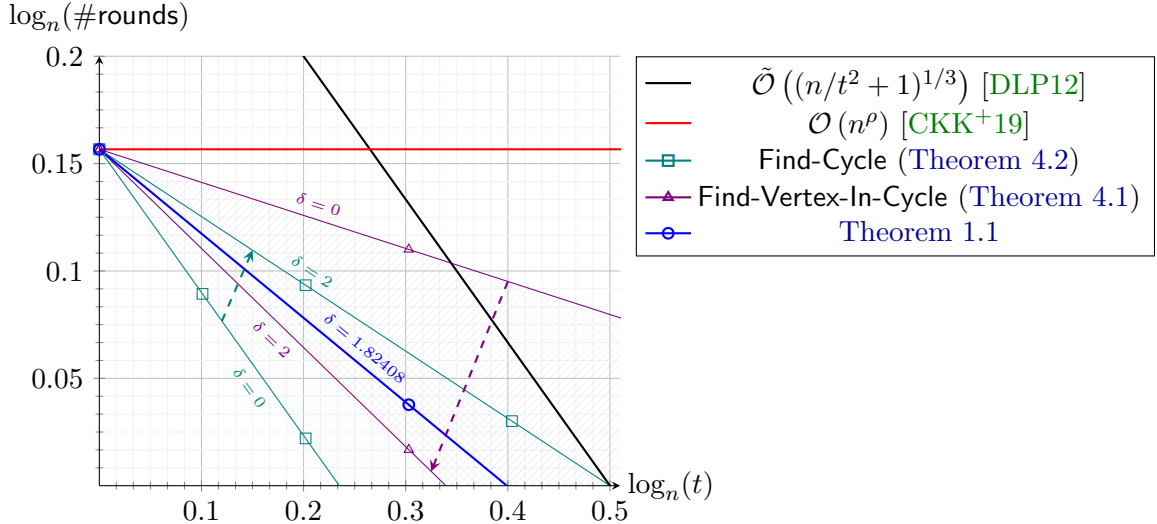
\begin{figure}[ht]
    \center
    \begin{tikzpicture}
        \begin{axis}[
                axis lines=middle,
                xlabel={$\log_n(t)$},
                ylabel={$\log_n(\#\mathsf{rounds})$},
                xlabel style={at={(ticklabel* cs:1.2)}, anchor=east},
                ylabel style={at={(ticklabel* cs:1.15)}, anchor=north},
                xmin=0, xmax=0.51,
                ymin=-0.0, ymax=0.2, %
                domain=0:10,
                samples=100,
                legend pos=outer north east,
                grid=both, %
                grid style={line width=.1pt, draw=gray!10}, %
                major grid style={line width=.2pt,draw=gray!50}, %
                minor tick num=5, %
                yticklabel style={/pgf/number format/fixed}
            ]

            \def\offsetlen{.8cm}
            \def\offsetsep{.15cm}
            \def\overshoot{.6cm}
            \def\ALPHA{0.3139}
            \newcommand{\middeltaval}{\dt}
            \def\colA{teal}
            \def\colC{violet}

            \addlegendimage{black, thick}
            \addlegendentry{$\TO{(n/t^2+1)^{1/3}}$ \cite{dolev2012tri}};
            \addlegendimage{red, thick}
            \addlegendentry{$\BO{\nr}$  \cite{Censor-HillelKK19}};
            
            \addlegendimage{\colA, thick, mark=square}
            \addlegendentry{\algA (\Cref{thm:h-cycle induced})};
            \addlegendimage{\colC, thick, mark=triangle}
            \addlegendentry{\algC (\Cref{thm:h-cycle})};

            \addlegendimage{blue, thick, mark=o}
            \addlegendentry{\Cref{thm:main in t}};

            \def\R{\rhoval} %
            \def\B{\betaval} %

            \addplot[name path=f0, \colA, thin, mark=square, mark repeat=1] {(\R)-x*(2+0*(\R-1))/(3-0)} node[pos=0.016, sloped, below] {\tiny{$\delta=0$}};
            \addplot[name path=f2, \colA, thin, mark=square, mark repeat=2] {(\R)-x*(2+2*(\R-1))/(3-2)} node[pos=0.02, sloped, above] {\tiny{$\delta=2$}};

            \addplot[name path=g0, \colC, thin, mark=triangle, mark repeat=3] {(\R)-x*\R/((1-\B)*(3-0))} node[pos=0.021, sloped, above] {\tiny{$\delta=0$}};
            \addplot[name path=g2, \colC, thin, mark=triangle, mark repeat=3] {(\R)-x*\R/((1-\B)*(3-2))} node[pos=0.018, sloped, below] {\tiny{$\delta=2$}};

            \addplot[name path=DLP, black, thick] {1/3-2*x/3};
            \addplot[name path=MM, red, thick] {\R};
            \addplot[name path=meet, blue, thick, mark=o, mark repeat=3] {(\R)-x*\R/((1-\B)*(3-\middeltaval))} node[pos=0.025,sloped,above] {\tiny$\delta=\middeltaval$};

            \tikzfillbetween[of=f0 and f2]{\colA, opacity=0.22, pattern=north west lines, pattern color=\colA};

            \tikzfillbetween[of=g0 and g2]{\colC, opacity=0.22, pattern=north east lines, pattern color=\colC};

            \def\axa{0.12}
            \def\axb{0.15}
            \draw[->, >=stealth, \colA, thick, dashed] (axis cs:\axa,{(\R)-\axa*(2+0*(\R-1))/(3-0)}) -- (axis cs:\axb,{(\R)-\axb*(2+2*(\R-1))/(3-2)});
            \def\bxa{0.4}
            \def\bxb{0.15} 
            \draw[->, >=stealth, \colC, thick, dashed] (axis cs:\bxa,{(\R)-\bxa*\R/((1-\B)*(3-0))}) -- (axis cs:{\bxb*(1-\B)/\R},{(\R)-(\bxb*(1-\B)/\R)*\R/((1-\B)*(3-2))});

        \end{axis}
    \end{tikzpicture}
    \caption{
        An illustrative comparison between our results and prior work, for the case of triangles. For each algorithm, we plot the base-$n$ logarithm of the number of rounds as a function of the base-$n$ logarithm of the number of triangles. An additional axis represents the value of $\delta$ ranging from 0 to 2. For a fixed $t$, \algA performs faster as $\delta$ decreases, with its round complexity depicted by the area shaded in teal. Conversely, \algC performs better as $\delta$ increases, and its round complexity is shown by the area shaded in violet.}
        \label[figure]{fig2:comp}
\end{figure}

Before presenting the algorithms, we overview the color-coding technique \cite{alon1995color}, which is a common method used to find paths or cycles of constant length $h$.
To detect an $h$-cycle, first color the vertices of the graph using $h$ colors, where each vertex is colored uniformly at random and independently of all other vertices. Then look for a \emph{colorful} $h$-cycle, which is an $h$-cycle with exactly one vertex of each color.
This provides additional structure, which a detection algorithm can benefit from.
However, not every coloring induces a colorful $h$-cycle, which means that to detect an $h$-cycle, we might have to repeat this experiment multiple times, until we sample a coloring that induces a colorful $h$-cycle.

In more detail, given a graph $G$ with $n$ vertices, we sample a uniformly random coloring $\varphi:V\to[h]$, which means that $\varphi$ colors each vertex uniformly independently at random. We then build the auxiliary graph $G_\varphi$ which is a directed graph, as follows.
\begin{definition}[$G_\varphi$]\label{def:aux coloring}
    \sloppy{Given a graph $G$, and a coloring $\varphi:V\to[h]$, we define a new directed graph $G_\varphi$ on the same vertex set, with a set of directed edges $E(G_\varphi)\triangleq\set{(u,v)\in E(G)\mid \varphi(v)=(\varphi(u) + 1) \mod h}$. That is, only a subset of the edges is kept, and it consists of the edges from the vertices of color $i$ to the vertices of color $(i+ 1)\mod h$, for $i\in[h]$.}
\end{definition}
The graph $G_\varphi$ has the property that every walk of length smaller than $h$ is a simple path, and every closed walk of length $h$ is a cycle. Here, a \emph{walk} of length $h$ on a (directed) graph is a sequence of vertices $(v_1,v_2,\ldots, v_{h+1})$ not necessarily distinct, such that for $i\in[h]$ we have that $(v_i,v_{i+1})$ is an edge in $G$.
We say that a walk is a simple path if all the vertices in the walk are distinct. A walk $(v_1,v_2,\ldots, v_h,v_{h+1})$ is \emph{closed} if $v_1=v_{h+1}$.

We explain how we benefit from the property that every closed walk of length $h$ in $G_\varphi$ is a cycle.
Let $A_\varphi$ denote the adjacency matrix of $G_\varphi$.
We can compute the $h$-th power of the matrix $A_\varphi$, and look at the diagonal of the obtained matrix. Then, $G_\varphi$ is $h$-cycle free if and only if all the entries on the diagonal are equal to $0$.
Clearly, if $G$ does not contain an $h$-cycle, then for any coloring $\varphi$, we have that $G_\varphi$ does not contain an $h$-cycle.
The more interesting property of this random coloring is that if $G$ contains an $h$-cycle, then the probability that $G_\varphi$ contains one is at least $1/h^h$, as we prove next.
\begin{claim}\label{claim:simple color-coding}
    Let $G$ be a graph with at least one $h$-cycle.
    Let $\varphi:V\to[h]$ be some uniformly random coloring.
    Then $G_\varphi$ contains an $h$-cycle with probability at least $\frac{1}{h^h}$.
\end{claim}
\begin{proof}
    Let $C=(v_1,\ldots, v_h,v_1)$ be some $h$-cycle in $G$.
    We prove that $C$ is also an $h$-cycle in $G_\varphi$ with probability at least $\frac{1}{h^h}$.
    For $C$ to be an $h$-cycle in $G_\varphi$, the vertex $v_i$ must be colored by the color $i$, for $i\in[h]$. As each vertex is colored with a uniformly random color, this event occurs with probability $1/h^h$, which proves the claim.
\end{proof}

\subsection{The Algorithm \algC}\label{ssec:c1}
We explain how to detect an $h$-cycle in time $\BO{\MM{n,n,n/x}\cdot \log^2 n}$ \whp, with a one-sided error, as stated in the next theorem.
\begin{theorem}\label{thm:h-cycle}
    \sloppy{There exists a randomized \CC algorithm to detect an $h$-cycle in time $\TO{\MM{n,n,\frac{n}{x}}}$ \whp, with a one-sided error.}
\end{theorem}

Let $G$ be a graph with $n$ vertices and $t$ copies of an $h$-cycle, for a fixed constant $h$.
For a graph $H$, we denote by $\Vx(H)$ the set of vertices that participate in an $h$-cycle in $H$. %
Let $x=\abs{\Vx(G)}$.
We prove \Cref{thm:h-cycle} by
analyzing the following random process. We implement it in the \CC model in \Cref{ssec:implement algc}.

~\\\textbf{\algC .}
The input of the algorithm is a graph $G$ and some value $p\in[0,1]$. The output is \true if at least one $h$-cycle is detected, and \false otherwise.
The algorithm works as follows.
The algorithm samples a coloring $\varphi:V(G)\to[h]$, uniformly at random, and uses it to define a new auxiliary graph $\Gp$, as explained in \Cref{def:aux coloring}.
Let $V_i$ denote the set of vertices in $\Gp$ that are assigned the color $i$, for $i\in[h]$.
The algorithm then samples a subset of vertices from $V_1$, by sampling each vertex independently with probability $p$. Let $U_1$ denote the set obtained.
Define $\Hp$ as the induced subgraph of $\Gp$ with the vertex set $U_1\cup \bigcup_{i=2}^h V_i$.
Let $A_{\Hp}$ denote the adjacency matrix of the graph $\Hp$.
Next, the algorithm exactly counts the number of $h$-cycles in $\Hp$ using rectangular matrix multiplication.
That is, it computes the trace of the $h$-th power of $A_{\Hp}$, and outputs \true if it is not zero, and \false otherwise.
Clearly, this can be computed by first computing the $h$-th power of $A_{\Hp}$, and then computing its trace, which takes $\BO{\MM{n,n,n}}$ rounds.
However, a faster well-known way to compute this trace, without computing the $h$-th power of $A_{\Hp}$, is as follows.
Compute the following product:
\begin{align*}
    A_{\Hp}[U_1,V_2]\cdot A_{\Hp}[V_2,V_3]\cdots A_{\Hp}[V_{h-1},V_{h}]\cdot A_{\Hp}[V_h,U_1]\;,
\end{align*}
where for $S,T\subseteq V$ the matrix $A_{\Hp}[S,T]$ denotes the rectangular matrix with $\abs{S}$ rows and $\abs{T}$ columns, every for every $s\in S$ and $T\in t$ we have that $(A_{\Hp}[S,T])_{s,t}=1$ if $(s,t)\in E(\Hp)$ and $0$ otherwise. This matrix is also called the biadjacency matrix.
The order in which the multiplications are computed affects the round complexity.
The algorithm computes this product by sequentially multiplying a rectangular matrix of size at most $4np\times n$ and a matrix of size at most $n\times n$, to get a new matrix of size $4np\times n$. In other words, the algorithm first computes the product $A_{\Hp}[U_1,V_2]\cdot A_{\Hp}[V_2,V_3]$, to obtain some matrix $B_2$, and then computes the product $B_2\cdot A_{\Hp}[V_3,V_4]$.
In this way, the algorithm does not multiply two square matrices of size $n$, and can benefit from the fact that it only computes the product of one smaller rectangular matrix with a square one.
This completes the description of the algorithm.

Clearly, the algorithm never outputs \true if the graph $G$ is $h$-cycle free. In what follows, we give a lower bound on the probability that it outputs \true when the graph has $h$-cycles.
\begin{claim}\label{claimz:1}
    If the sampling probability of vertices from $V_1$ into $U_1$ satisfies $p\geq \frac{4 h^h}{x}$, then the algorithm outputs \true with probability at least $\frac{1}{4h^h}$.
\end{claim}
\begin{proof}[Proof of \Cref{claimz:1}]
    Let $q\triangleq\frac1{h^h}$.
    Let $S\triangleq V_1\cap \Vx(\Gp)$, which is the set of vertices that participate in an $h$-cycle in $\Gp$, and also belong to the first color class.
    Let $S'\triangleq S\cap U_1$, which is the subset of vertices from $S$, which are also sampled into $\Hp$.
    We show that
    \begin{enumerate}
        \item $\abs{S}\geq x\cdot q/2$ with probability at least $q/2$, and
        \item $\abs{S'}$ has the distribution of a binomial random variable with $\abs{S}$ trials (note that the number of trials is also a random variable), and success probability $p$.
    \end{enumerate}
    Using conditional probability and the fact $p\geq 4/(xq)$, we get that with probability at least $q/2$, the set $S'$ is not empty, which means that the subgraph $\Hp$ contains an $h$-cycle, and therefore the algorithm outputs \true.
    We first prove the first item, where the randomness is over the choice of the coloring $\varphi$.
    For $v\in V$, let $X_v$ denote the indicator random variable for the event that $v\in S$. Note that $\abs{S}=\sum_{v\in V} X_v=\sum_{v\in \Vx(G)}X_v$, and therefore, by the linearity of expectation, it suffices to show that for every $v\in\Vx(G)$, we have $\Exp{X_v}\geq q$.
    To see this, fix some vertex $v\in \Vx(G)$ and an $h$-cycle $C=(v,u_2,\ldots, u_h,v)$ that contains $v$. By \Cref{claim:simple color-coding}, we know that $C$ is also an $h$-cycle in $\Gp$ with probability $q$, and therefore, $v\in \Vx(\Gp)$ with probability at least $q$, as $v$ could belong to more than one $h$-cycle in $G$.
    We get that $\Exp{\abs{S}}\geq x\cdot q$.

    We can now use a reverse Markov's inequality (\Cref{thm:rev mark}), to get that with probability at least $q/2$, we have that $\abs{S}\geq \Exp{\abs{S}}/2$:
    \begin{align*}
        \Pr{\abs{S}\geq \Exp{\abs{S}}/2}
        \geq \frac{\Exp{\abs{S}}-\Exp{\abs{S}}/2}{x-\Exp{\abs{S}}/2}
        \geq \frac{\Exp{\abs{S}}-\Exp{\abs{S}}/2}{\Exp{\abs{S}}/q-\Exp{\abs{S}}/2}
        \geq q/2\;.
    \end{align*}
    To conclude, we showed that if $p\geq \frac{4h^h}{x}=\frac{4}{q\cdot x}$, then $\Pr{\abs{S}\geq \Exp{\abs{S}}/2}\geq q/2$, and since $\Exp{\abs{S}}\geq xq$, we get that $\begin{aligned}
            \Pr{\abs{S}\geq xq/2}\geq q/2
        \end{aligned}$,
    which proves the first item.

    We now move to the second part of the proof, which establishes a lower bound on $\abs{S'}$.
    Let $Y$ denote the random variable equal to $\abs{S'}$, which is the number of vertices in $U_1$ which participate in $\Hp$.
    For $v\in V_1$ let $Y_v$ denote the indicator random variable for the event that $v$ is in the set $\Vx(\Hp)$.
    The key point here is that the set $V_1$ is an independent set over $\Hp$, and therefore the random variables $\set{Y_v}_{v\in V_1}$ are pairwise independent.
    This means that for every $v\in S$, the event $\set{Y_v=1}$ is equivalent to the event $\set{v\in U_1}$, which occurs with probability $p$.
    Note that for $v\in V_1\setminus \Vx(\Gp)$, we have that $Y_v$ is always $0$.
    We get that $Y$ has the distribution of a binomial random variable with $\abs{S}$ trials and success probability $p$. We can now prove that
    \begin{align*}
        \Pr{Y\neq 0}\geq \Pr{Y\neq 0\mid \abs{S}\geq 1/p}\cdot \Pr{\abs{S}\geq 1/p}\geq \brak{1-(1-p)^{1/p}}\cdot \frac{q}{2}\geq \frac{q}{4}\;.
    \end{align*}
    This means that the probability that $\Hp$ contains an $h$-cycle is at least $q/4=\frac{1}{4h^h}$, and when this occurs the algorithm outputs \true,
    which completes the proof.
\end{proof}

\subsection{Implementation of \algC in the \CC Model}
\label{ssec:implement algc}
In this subsection, we explain how to implement the algorithm specified in the previous subsection in the \CC model, and prove \Cref{thm:h-cycle}.
The implementation of the algorithm is divided into two parts. In the first part, the graph $\Hp$ is sampled and every vertex learns its edges in the graph and the color of every other vertex in the graph.
In the second part, the algorithm computes the exact number of $h$-cycles in $\Hp$ using multiple rectangular matrix multiplications, and outputs \true if this number is positive and \false otherwise.

We implement the first part as follows:
\begin{enumerate}
    \item A leader vertex $v_0$ locally samples a coloring $\varphi:V\to [h]$ uniformly at random, and a subset of vertices $U_1\subseteq V_1$, where each vertex from $V_1$ is added to $U_1$ independently with probability $p$.
          If $\abs{U_1}> 4np$ then $v_0$ outputs \false, and all other vertices also output \false.
    \item
          The leader vertex $v_0$ then notifies each vertex of its color under $\varphi$, and whether it is in $U_1$ or not.
    \item Each vertex $u$ notifies each other vertex $v$ of its color, and whether it is in $U_1$ or not.
\end{enumerate}
This whole part can be implemented in $\BO{1}$ rounds, as in every round each vertex sends and receives at most $\BO{n}$ messages of $\BO{1}$ bits each.
This completes the implementation of the first part of the algorithm.

In what follows, we implement the second part of the algorithm.
Let $A_{\Hp}$ denote the adjacency matrix of the graph $\Hp$.
To implement the second part, the algorithm computes the product
$\prod_{i=1}^h A_i$, where the matrices $\set{A_i}_{i\in[h]}$ are defined as follows:
\begin{align*}
    \forall i\in[h], \quad A_i\triangleq\begin{cases}
        A_{\Hp}[U_1,V_2],     & \text{if } i=1     \\
        A_{\Hp}[V_h,U_1],     & \text{if } i=h     \\
        A_{\Hp}[V_i,V_{i+1}], & \text{otherwise. }
    \end{cases}
\end{align*}
The product $\prod_{i=1}^h A_i$ is computed recursively.
We compute $B_{i+1}=B_i\cdot A_{i+1}$, for $i\in\set{1,\ldots, h}$, where $B_1=A_1$.
After computing $B_h$, which is a square matrix of size $\abs{U_1}$, each vertex $u\in U_1$ outputs \true if the entry $B_h[u,u]$ is positive and outputs \false otherwise. Any vertex not in $U_1$ outputs \false.
This completes the implementation of the second part, and therefore, the entire description of the implementation of the algorithm.

The following claim bounds the running time of the algorithm \algC.
\begin{claim}\label{claim:runtime vit}
    The round complexity of the algorithm \algC$(G,p)$ is $\BO{\RM{4np}}$.
\end{claim}
\begin{proof}
    We claim that computing $B_{i+1}$ given $B_i$ and $A_{i+1}$ can be done in $\RM{\abs{U_1}}$ rounds.
    Since we need to compute $h$ products, overall we get an $\BO{\RM{\abs{U_1}}}$ rounds algorithm, as $h=\BO{1}$.
    We want to reduce the problem of computing the product of the matrices $(B_i,A_{i+1})$ to the problem $\RM{(B_i,A_{i+1})}$. To do that, we need to provide an algorithm to route the input such that the $j$-th vertex holds the $j$-th column of the matrix $B_i$, and the $j$-th column of the matrix $A_{i+1}$. This can be implemented in $\BO{1}$ rounds, using \Cref{lemma:routing}, because each vertex needs to send at most one column of $n$ entries to one other vertex, and similarly each vertex needs to receive at most one column from one other vertex.
\end{proof}

Using \Cref{claimz:1}, we get the following conclusion that addresses the probability that the algorithm detects an $h$-cycle.
\begin{conclusion}\label{conclusionz:y}
    If $p\geq 4h^h/x$ then the algorithm \algC$(G,p)$ outputs \true with probability at least $\frac{1}{8h^h}$.
\end{conclusion}
\begin{proof}
    Let $\mathcal{B}$ denote the bad event that the leader vertex outputs \false in the first round, which occurs because the set $U_1$ is larger than $4np$.
    Note that $\Exp{\abs{U_1}}=np/h$.  Therefore, by Markov's inequality, the probability that the size of the set $U_1$ exceeds $4np$ is at most $\frac{1}{4}$.

    To complete the proof, we use \Cref{claimz:1}, which states that if $p \geq 4h^h/x$, then the algorithm \algC$(G,p)$ outputs \true with a probability of at least $\frac{1}{4h^h}$. Overall, the probability that $\Hp$ contains an $h$-cycle and that the event $\mathcal{B}$ does not occur, is at least $\frac{1}{8h^h}$, thus completing the proof.
\end{proof}

We can now prove \Cref{thm:h-cycle} using \Cref{conclusionz:y} and a doubling algorithm.
\begin{proof}[Proof of \Cref{thm:h-cycle}]
    We begin with the description of a doubling algorithm.
    \begin{enumerate}
        \item Start with an initial guess $p=1/n$.
        \item Execute the algorithm \algC$(G,p)$ for $80h^h\cdot \log n$ times (sequentially).
        \item If one of these executions detects an $h$-cycle return \true and terminate.
        \item Otherwise, update the value of $p$ by setting $p\gets 2p$, and go to Step $2$. If $p=1$ return \false.
    \end{enumerate}
    Let $i^*$ be the minimal index for which $p\geq 4h^h/x$.
    That is, $i^*=\ceil*{\log(4h^h \cdot n/x)}$.
    In the $i^*$-th iteration, by \Cref{conclusionz:y} each execution of the algorithm has a probability of at least $\frac1{8h^h}$ of detecting an $h$-cycle, which results in the termination of the algorithm.
    Since different executions are independent, the probability that no execution detects an $h$-cycle is
    \begin{align*}
        \brak{1-\frac1{8h^h}}^{80h^h\cdot \log n}\leq \exp(-10\log n)\leq \frac{1}{n^{10}}\;.
    \end{align*}
    This means that the algorithm makes at most $\ceil*{\log(n)}$ iterations before termination.
    By \Cref{claim:runtime vit}, the $i$-th iteration takes $\TO{\RM{\frac{n}{2^i}}}$ rounds, and therefore the total running time of the algorithm is
    \begin{align*}
        \sum_{i=0}^{\ceil*{\log(n)}} \TO{\RM{\frac{n}{2^i}}} = \TO{\RM{\frac{n}{x}}}\;,
    \end{align*}
    with probability at least $n^{-10}$, which completes the proof.
\end{proof}

\subsection{The Algorithm \algA}\label{ssec:c3}
In the next two subsections, we explain how to prove the following theorem.
\begin{theorem}\label{thm:h-cycle induced}
    There exists a randomized \CC algorithm to detect an $h$-cycle in time $\TO{\MM{\frac{n}{x},\frac{n}{x},\frac{n}{x};x^\delta}}$ \whp, with one-sided error.
\end{theorem}

Recall that $G$ is a graph with $n$ vertices and $t$ copies of an $h$-cycle for $h=\BO{1}$. For a graph $H$, we denote by $\Vx(H)$ the set of vertices that participates in an $h$-cycle in $H$.
Let $x=\abs{\Vx(G)}$. We also use $\delta$ for the solution for $x^{h-\delta}=2ht$ satisfies that $\delta\in[0,h-1]$.
\begin{remark}
    Recall that $\MM{\frac{n}{x},\frac{n}{x},\frac{n}{x};x^\delta}=\BO{\nr\cdot x^{-(2+\delta(\rho-1))}}$, by \Cref{thm:p balanced MM}.
\end{remark}

We prove \Cref{thm:h-cycle induced} by
analyzing the following random process.
We implement it in the \CC model in \Cref{ssec:implement alga}.

\newcommand*{\Guvf}{F_\varphi}

~\\\textbf{\algA .}
The input of the algorithm is a graph $G$ A graph $G$, a value $p\in[0,1]$, and a value $a\in[0,2]$. 
The output is \true if at least one $h$-cycle is detected, and \false otherwise.
The algorithm works as follows.
\begin{enumerate}
    \item Sample uniformly at random a coloring $\varphi:V\to[h]$.
    \item Sample $r\gets 8(4h)^{h+2}\cdot p^{-a}$ subsets of vertices $\UU=(U_1,\ldots, U_r)$, where each vertex joins $U_i$ independently with probability $p$ for $i\in[r]$.
    \item For every $U\in \UU$, define two graphs. The first one is the induced graph $F=G[U]$, and the second one is the colored directed graph $F_\varphi$, obtained from applying $\varphi$ on $F$. Denote the adjacency matrix of $F_\varphi$ by $M_U$.
    \item For $U\in \UU$, compute the trace of the $h$-th power of the matrix $M_U$, and output \true if for at least one set $U$, this trace is not zero. Otherwise, output \false.
\end{enumerate}
Fix some set $U\in\UU$,
and a random coloring $\varphi$, and let $F=G[U]$, and $\Guvf=(G[U])_\varphi$.
We prove that for $p\geq 1/x$, the subgraph $\Guvf$ contains an $h$-cycle with probability $\Omc[x^{-\delta}]$.
For that, it suffices to prove that if $p\geq 1/x$ then $F$ contains an $h$-cycle with probability at least $\frac{1}{x^{\delta}\cdot (4h)^{h+1}}$: We proved in \Cref{claim:simple color-coding} that if $F$ contains an $h$-cycle then $\Guvf$ contains an $h$-cycle with probability at least $\frac{1}{h^h}$.
In the next proposition, we prove that the subgraph $F$ contains an $h$-cycle with probability at least
$\frac{1}{x^{\delta}\cdot (4h)^{h+2}}$, if $p\geq \frac{1}{x}$.

\begin{proposition}\label{claimz:2}
    If $p\geq \frac1{x}$, then $F$ contains an $h$-cycle with probability at least $\frac{1}{x^{\delta}\cdot (4h)^{h+2}}$.
\end{proposition}
To prove the above, we use the second moment method \cite[Theorem 4.3.1]{alon2016probabilistic}.
Let $t_F$ denote the random variable equal to the number of copies of an $h$-cycle in the subgraph $F$.
We need to give an upper bound on $\Exp{{t_F}^2}$, which closely relates to the following definition.
\begin{definition}
    For $\ell\in [h+1,2h]$ , let $\FH^\ell$
    denote the number of pairs of copies of an $h$-cycle in $G$ which share exactly $2h-i$ vertices, or equivalently, which the union of their vertex sets contains exactly $i$ vertices.
\end{definition}
Let $(C_1,\ldots, C_t)$ denote the set of $h$-cycles in $G$, and let $X_i$ be an indicator random variable for the event that $C_i$ is also in $F$.
Note that $t_F=\sum_{i=1}^{t}X_i$, and therefore $\Exp{t_F^2}=\sum_{i,j\in[t]} \Exp{X_i\cdot X_j}$.
The following claim will help us obtain an upper bound on $\Exp{t_F^2}$.
\begin{claim}\label{claimz:H share i vertices}
    For every
    $i\in[h+1,2h-1]$, it holds that $\FH^i\leq (2h)^{h}\cdot t\cdot x^{i-h}\;.$
\end{claim}
\begin{proof}[Proof of \Cref{claimz:H share i vertices}]
    For $i\in [h+1,2h-1]$, every pair of $h$-cycles $(C',C'')$ %
    that share $2h-i$ vertices can be encoded as follows.
    First, we specify the index of $C'$ among the $t$ options.
    Then, we specify the vertices in $V(C'')\setminus V(C')$, which have at most $x^{i-h}$ possible options.
    Then, we specify the edges of $C''$ between the already specified vertices, which have at most $(2h)^{h}$ options.
\end{proof}
We can now apply the second moment method:
\begin{align*}
    \Pr{t_F> 0}\geq \frac{\Exp{t_F}^2}{\Exp{t_F^2}}\;,
\end{align*}
where $\Exp{t_F}=t\cdot p^h$, to prove \Cref{claimz:2}.
\begin{proof}[Proof of \Cref{claimz:2}]
    We first use \Cref{claimz:H share i vertices} to bound $\Exp{t_F^2}$. We have
    \begin{align*}
        \Exp{t_F^2} & =\sum_{i,j\in[t]} \Exp{X_i\cdot X_j}                                                \\
                    & = t\cdot p^h  + \sum_{i=h+1}^{2h-1}\FH^{i}\cdot p^{i} + (t\cdot p^h)^2              \\
                    & \leq t\cdot p^h + ((2h)^{h}\cdot t)\sum_{i=1}^{h-1}x^{h-i}p^{2h-i} + (t\cdot p^h)^2 \\
                    & = t\cdot p^h + p^h\cdot ((2h)^{h}\cdot t)\sum_{i=1}^{h-1}(xp)^{i} + (t\cdot p^h)^2  \\
                    & \leq tp^h\cdot \brak{1 + (2h)^h\cdot h(1+(xp)^{h-1}) + tp^h}                        \\
                    & = tp^h\cdot \brak{1 + (2h)^h\cdot h(1+(xp)^{h-1}) + (xp)^h\cdot x^{-\delta}/2h}\;,
    \end{align*}
    where the penultimate inequality follows as $\sum_{i=1}^{h-1}(xp)^{i}\leq h(1+(xp)^{h-1})$, and the last equality uses $x^{h-\delta}=2h\cdot t$, which follows from the definition of $\delta$.
    We can further upper bound this expression by
    \begin{align*}
        (2h)^{h+1}\cdot tp^h\cdot \brak{(xp)^{h-1} + (xp)^h\cdot x^{-\delta}}\;.
    \end{align*}

    From \Cref{claimz:H share i vertices} and the second moment method, we have
    \begin{align*}
        \Pr{t_F> 0}
        \geq \frac{\Exp{t_F}^2}{\Exp{t_F^2}}
         & \geq  \frac{(t p^h)^2}{(2h)^{h+1}\cdot tp^h\cdot \brak{(xp)^{h-1} + (xp)^h\cdot x^{-\delta}}}  \\
         & = \frac{1}{(2h)^{h+2}}\cdot \frac{x^{h-\delta}\cdot p^h}{(xp)^{h-1} + (xp)^h\cdot x^{-\delta}} \\
         & = \frac{1}{(2h)^{h+2}}\cdot \frac{x^{-\delta}}{(xp)^{-1} + x^{-\delta}}
        \;.
    \end{align*}
    The inequalities hold due to the following reasons.
    The first inequality follows from the second moment method.
    The second inequality follows by plugging in the previously computed values for $\Exp{t_F}^2$ and $\Exp{t_F^2}$.
    The penultimate equality follows by canceling terms and plugging in $x^{h-\delta}=2h t$.
    The last equality is obtained by simply canceling terms.

    For $p\geq 1/x$, we get that $(xp)^{-1}\leq 1$, and therefore $(xp)^{-1} + x^{-\delta}\leq 2$, which proves that
    $\frac{x^{-\delta}}{(xp)^{-1} + x^{-\delta}}\geq \frac{1}{2x^{-\delta}}$.
    Therefore, $\Pr{t_F> 0}\geq \frac{1}{(2h)^{h+2}}\cdot \frac{1}{2x^{\delta}}$, which completes the proof of \Cref{claimz:2}.
\end{proof}
The following lemma shows that the described experiment outputs \true with probability at least $\frac{1}{2h^h}$.
\begin{lemma}\label{lemma:GU plus vf}
    Sample a uniformly random coloring $\varphi:V\to[h]$.
    Sample $r$ subsets of vertices $\UU=(U_1,\ldots, U_r)$, where each vertex joins $U_i$ independently with probability $p$.
    For $p\geq 1/x$, and $r\geq 8\cdot (4h)^{h+2}\cdot x^\delta$, the probability that there exists $U\in\UU$ such that $(G[U])_\varphi$ contains an $h$-cycle is at least $\frac{1}{2h^h}$.
\end{lemma}
\begin{proof}
    We first show that the probability that there exists $U\in\UU$ for which $F\triangleq G[U]$ contains an $h$-cycle is at least $1/2$.
    We then reveal the random coloring $\varphi$ and show that $F_\varphi$ contains an $h$-cycles with probability at least $\frac{1}{2h^h}$, which follows by \Cref{claim:simple color-coding}.

    By \Cref{claimz:2}, if $p\geq 1/x$, then for each $U_i\in \UU$, we have that $G[U_i]$ contains an $h$-cycle with probability at least $q'\triangleq\frac{1}{(4h)^{h+2}\cdot x^{\delta}}$.
    Note that $r=8/q'$.
    Since the samples are pairwise independent, the probability that for every $i\in[r]$ the subgraph $G[U_i]$ does not contain an $h$-cycle is at most $(1-q)^r=(1-q)^{8/q}\leq \frac1{e^8}\leq 1/2$.

\end{proof}

\subsection{Implementation of \algA in the \CC Model}\label{ssec:implement alga}
In this subsection, we explain how to implement the algorithm specified in the previous subsection in the \CC model, and prove \Cref{thm:h-cycle induced}.
The algorithm takes as input two parameters, $p\in[1/\sqrt{n},1]$ and $a\in[0,2]$.
The implementation of the algorithm is divided into two parts. In the first part, a single coloring $\varphi:V\to[h]$ is sampled, along with $s$ random induced subgraph. Additionally, the vertices are partitioned into $s$ sets, each with $n/s$ vertices. We call each such set a \emph{team}.
Each team is responsible for checking if one of the sampled subgraph contains an $h$-cycle.
In the second part, each team computes the exact number of $h$-cycles in the subgraph it is responsible for using matrix multiplication.

~\\We implement the first part as follows:
\begin{enumerate}
    \item Each vertex samples a color from $[h]$ uniformly at random, and notifies all other vertices of its choice.
    \item The vertices are partitioned into $s$ sets $(V_1,\dots,V_s)$ called \emph{teams}, each of size $n/s$, where $s= 8(4h)^{h+2}\cdot p^{-a}$.
    \item Define $s$ subsets of vertices, $\UU\triangleq\brak{U_1,\dots, U_{s}}$. Each vertex joins $U_i$ independently with probability $p$, for $i\in[s]$.
    \item If $\UU$ is not $p$-balanced, all vertices output \false, and the algorithm terminates.
    \item Every vertex $v$ learns the IDs of all vertices in the set $U_j$ for all $j$ such that $v\in U_j$.
\end{enumerate}

This can be implemented in $\BO{\log n}$ rounds:
The first step can be implemented in $\BO{1}$ rounds.
The second step in which every vertex is assigned to one team can be implemented in $\BO{1}$ rounds. The third step does not require any communication.
The fourth step can be implemented in $\BO{1}$ rounds by \Cref{claim:no-balance}, and the last step can be implemented in $\BO{\log n}$ rounds using \Cref{prop2:prob-ind}.
This completes the implementation of the first part of the algorithm.

In the second part of the algorithm, the $i$-th team is responsible for exactly counting the number of $h$ cycles in the graph $(G[U_i])_\varphi$ for $i\in[s]$.
\sloppy{This can be implemented in $\BO{\nr\cdot p^{2+a(\rho-1)}+\log  n}$ rounds by \Cref{thm:p balanced MM}, so that each team computes the $h$-th power of the adjacency matrix of the subgraph it is responsible for.}
Then, each team can compute the trace of that matrix in $\BO{1}$ rounds. If the trace is not zero, the team leader outputs \true, and otherwise outputs \false. All other vertices in the team output \false.
This completes the implementation of the second part of the algorithm.

The following conclusion summarize the round complexity of the algorithm.
\begin{conclusion}
    Assume that $\tfrac{1}{\sqrt{n}}\leq p\leq 1$ and $a\leq 2$.
    Then, $\malgA{p,a;G}$ completes in $\BO{\nr\cdot p^{2+a(\rho-1)}+ \log n}$ rounds (always).
\end{conclusion}

Next, we bound the probability that the algorithm fails to detect an $h$-cycle.
\begin{claim}
    \label{claimz:y7}
    \indent Assume $G$ contains at least one $h$-cycle. Then, for $p\geq 1/x$ and $a\geq \delta$, we have that
    $\malgA{p,a;G}$ outputs \true with probability at least $\frac{1}{3h^h}$.
\end{claim}
\begin{proof}
    The event that $\UU$ is $p$-balanced happens with probability at least $1-2/n^3$ by \Cref{claim:balance whp}.
    If the set $\UU$ is $p$-balanced, then $\malgA{p,a;G}$ outputs \true with probability at least $\frac{1}{2h^h}$, which follows immediately by \Cref{lemma:GU plus vf}.
    Therefore, $\malgA{p,a;G}$ outputs \true with probability at least
    $\frac{1}{2h^h} -2/n^3 \geq \frac{1}{3h^h}$.
\end{proof}
We are now ready to prove \Cref{thm:h-cycle induced}
\begin{proof}[Proof of \Cref{thm:h-cycle induced}]
    To complete the proof of \Cref{thm:h-cycle induced},  we need to explain how we overcome the problem of not knowing the values of $x$ and $\delta$. We use a doubling approach. Let $A_{i,j}$ denote an algorithm that executes algorithm $\malgA{p_i,a_j;G}$ in an infinite loop (each time with different randomness). We define $p_i=2^i/n$ for $i\in\zrn{\log n}$, and $a_j=j/\log n$ for $j\in\zrn{2\log n}$. Let $K=\set{(i,j)\mid i\in\zrn{\log n}\;, j\in\zrn{2\log n}}$, where $\abs{K}=\TO{1}$.
    We execute of the algorithms $\set{A_{i,j}}_{(i,j)\in K}$, by interleaving
    them. That is, we run the algorithm $A_{1,1}$ for a single round, then the algorithm $A_{1,2}$ for a single round, and so on.
    We stop as soon as one of the algorithms outputs \true.
    Next, we analyze the runtime of this ``lazy'' doubling approach.

    Let $i^*$ denote the minimal index for which $p_{i^*}\geq 1/x$ and let $j^*$ denote the first index for which $a_{j^*}\geq \delta$.
    Then, each execution of $A_{i^*,j^*}$ outputs \true, with probability at least $\frac1{3h^h}$, by \Cref{claimz:y7}.
    Therefore, after completing $16h^h\log n$ independent executions of $A_{i^*,j^*}$, the probability that no execution outputted \true, is at most $\brak{1-\frac1{2h^h}}^{16h^h\log n}\leq 1/n^{8}$.
    As each execution takes $\TO{\MM{\frac{n}{x},\frac{n}{x},\frac{n}{x};x^\delta}}$ rounds, we get that the total round complexity of the doubling algorithm is $\TO{K\cdot 16h^h\log n\cdot \MM{\frac{n}{x},\frac{n}{x},\frac{n}{x};x^\delta}}=\TO{\MM{\frac{n}{x},\frac{n}{x},\frac{n}{x};x^\delta}}$ \whp.
\end{proof}

\subsection{Wrap-Up: Fast Cycle Detection}\label{ssec:c5}

In this subsection, we wrap up to prove our fast algorithm for $h$-cycle detection, when $h=\BO{1}$, in both undirected and directed graphs. Our algorithm is the fastest for odd cycle detection when the number of cycles is super polylogarithmic, and for $h$-cycle detection in directed graphs, when the number of $h$-cycles is super polylogarithmic.
For graphs with small $t$, our algorithm has the same running time as the fastest algorithm for multiplying two matrices of size $n\times n$,
and our running time is never worse than it up to polylogarithmic factors.
\ThmMainT* %
\newcommand{\AR}{\mathcal{R}}
Let $\AR(G)$ denote the round complexity of \Cref{thm:main in t}.
Let $\AR_1(G),\AR_2(G)$ denote the round complexity of the algorithms in \Cref{thm:h-cycle induced} and \Cref{thm:h-cycle} respectively.
To prove \Cref{thm:main in t}, we show that for every graph $G$ with $n$ vertices and $t$ copies of an $h$-cycle, we have $\min\set{\AR_1(G),\AR_2(G)}\leq \AR(G)$.
To show that, we use a case analysis.
Recall that $x$ denotes the number of vertices in $G$ that participate in an $h$-cycle, and that $x^{h-\delta}=2ht$.
We show that if $\delta\geq \dt$, then $\AR_2(G)\leq \AR(G)$, and if
$\delta\leq \dt$, then $\AR_1(G)\leq \AR(G)$.

The theorem then follows, as we can run the algorithms \algA and \algC one step at a time, until one of them detects a triangle.
\begin{proof}[Proof of \Cref{thm:main in t}]
    
    ~\\\textbf{The Case $\delta\geq \dt$.}
    The execution of the algorithm \algC takes $\MM{\frac{n}{x},n,n}$ rounds, where $\MM{\frac{n}{x},n,n} \leq n^{\Bs}\cdot x^{-\As}$ by \Cref{def:line_parameters}.
    Since we assumed that $\delta\geq \dt$, we have $x=(2ht)^{1/(h-\delta)}\geq t^{1/(h-\dt)}$.
    We get that $\AR_2(G)\leq n^{\Bs}\cdot t^{-\frac{\As}{2-\dt}}
        \;.$
    By plugging in $\As=\aConst,\Bs=\bConst$, (see \Cref{CLAIM:NUMERIC RHO}) we get that the round complexity is bounded by
    $n^{\bConst}\cdot t^{-\frac{\aConst}{2-\dt}}$, which completes the proof of this case.

    ~\\\textbf{The Case $\delta\leq \dt$.}
    The execution of the algorithm \algC takes $\MM{\frac{n}{x},\frac{n}{x},\frac{n}{x};x^{\delta}}$ rounds, where
    \begin{align*}
        \MM{\frac{n}{x},\frac{n}{x},\frac{n}{x};x^{\delta}}
        =    & \BO{n^\rho/ x^{2+\delta(\rho-1)}}                  \\
        =    & \BO{n^\rho/ t^{\frac{2+\delta(\rho-1)}{h-\delta}}} \\
        \leq & \BO{n^\rho/ t^{\frac{2+\dt(\rho-1)}{h-\dt}}}       \\
        \leq & \BO{n^{\bConst}\cdot t^{\frac{\aConst}{h-\dt}}}
    \end{align*}
    The first equality follows from \Cref{thm2:mm runtime}.
    The penultimate inequality follows since
    the function $\delta\mapsto \frac{2+\delta(\rho-1)}{h-\delta}$ is monotonically decreasing in the range $\delta\in[0,\dt]$.
    The last inequality follows by setting $\rho\gets \bConst$,
    which completes the proof.

\end{proof}

\section{Quantum Algorithms for Triangle Detection}\label{sec:q}
In this section, we provide two \QCC algorithms for triangle detection.
We take a \CC algorithm $\AC$ for triangle detection with round complexity $\CRR$ and turn it into an \QCC algorithm $\AC'$ for triangle detection with round complexity $\CRR^{3/4}$.
As before, our algorithm samples a subgraph and looks for a triangle in it. Therefore, the algorithm never outputs \true if the graph is triangle-free.
We say the algorithm has a success probability of at least $\tau$ if it outputs \true with probability at least $\tau$.
By repeating the algorithm $\log n$ times, one can boost the success probability.
Our main theorem is as follows. 
\ThmQInCycle* %

As a warm-up, we prove the following special case.
\begin{theorem}\label{thm:quantum1}
    There exists a \QCC $\BO{n^{3\rho/4}}$-round algorithm for triangle detection, with a success probability of at least $1/2$.
\end{theorem}
~\\\textbf{Quantum computing background.}
In the \QCC model, during each round, every pair of vertices in the graph can exchange $\BO{\log n}$ qubits, similar to the standard \CC model where they exchange bits.
We follow the framework presented in \cite{le2018sublinear,IzumiG19,Censor-HillelFG22}: Given a universe $X$ and a function $f:X\to\set{0,1}$, we want to find an element $x\in X$ such that $f(x)=1$. We assume that for every $x\in X$, the leader vertex $v$ can evaluate $f(x)$ given $x$ in $r$ rounds. The algorithm uses Grover's search \cite{grover1996fast} to find an element $x$ such that $f(x)=1$.
In \cite{le2018sublinear} (see also \cite[Lemma 4]{Censor-HillelFG22}), the paper presents a quantum algorithm to find an element $x\in X$ such that $f(x)=1$, assuming such an element exists, in $\BOs{r\cdot \sqrt{\abs{X}}}$ rounds \whp in the \QCC model.

Besides the quantum search framework, our algorithm builds on a tool we developed in \Cref{sec:FMM}, which proves that in the \CC model, $n$ nodes can compute the product of a single square boolean matrix of size $np$ for $p\in[0,1]$ in $\BO{\nr\cdot p^2}$ rounds (\Cref{thm2:mm runtime}).

\subsection{Proof of \Cref{thm:q2}}
For the remainder of this section, we assume that the input graph contains at most $\TOs{t^{0.3992}}$ triangles. If it contains more, we can apply \Cref{thm:main in t} to determine in $\TO{1}$ rounds whether the graph has a triangle. 
This assumption is necessary because the proof of \Cref{thm:q2} involves partitioning the vertices into $t^2$ teams and utilizes \Cref{thm2:mm runtime}, which we can use only when the number of teams does not exceed $n$. Consequently, we require that $t \leq \sqrt{n}$.

The input for the algorithm is a graph $G$ with $n$ vertices. If $G$ contains a triangle, the algorithm outputs \true with probability at least $1/2$. Otherwise the algorithm output \false. The algorithm never outputs \true if the graph is triangle-free.

\begin{proof}[Proof of \Cref{thm:quantum1}]
    A leader vertex samples $\ell=8\log n/p^3$ subsets of vertices $(U_1,\ldots, U_\ell)$ uniformly at random, where each vertex is added to each set independently of all other vertices, with probability $p$. We choose the value of $p$ later.
    We emphasize that only the leader vertex is aware of the subsets.
    Our search space is $X=\set{G[U_i]}_{i\in[\ell]}$. We use $x_i$ to denote $G[U_i]$ for $i\in[\ell]$.
    The boolean function we evaluate $g:X\to\set{0,1}$ is defined as follows:
    \begin{align*}
        \forall x_i\in X\quad g(x_i)=\begin{cases}
            1 & \text{If $G[U_i]$ contains a triangle.} \\
            0 & \text{Otherwise.}
        \end{cases}
    \end{align*}
    Define $r$ as the maximal number of rounds required to evaluate $g(x)$ for some $x\in X$.
    By \cite{le2018sublinear}, the running time of the algorithm is $\BOs{r\sqrt{\abs{X}}}$.

    To prove the correctness of the algorithm, we show that if $G$ contains a triangle the algorithm outputs \true with a probability of at least $1/2$.
    To show that, we prove that if $G$ contains at least one triangle, then with probability at least $1-1/n^7$, there exists an index $i\in [\ell]$ such that $G[U_i]$ contains a triangle.
    If this is the case, then, according to the guarantees provided by the Grover search, the algorithm outputs \true \whp.

    Fix a triangle $S=(a,b,c)$ in $G$.
    Note that for every $i\in[\ell]$, we have that $\Pr{S\subseteq U_i}=p^3$.
    Since the sampled subsets $\set{U_i}_{i\in[\ell]}$ are independent, the probability that $S$ is in none of the samples is $(1-p^3)^\ell\leq\exp(-8\log n)=\frac1{n^8}$.
    Therefore, at least one sample contains a triangle with a probability of at least $1-1/n^8$, which proves the correctness of the algorithm.

    Next, we analyze the running time of the algorithm. We show that if $p=n^{-\rho/2}$, then $r=\BO{1}$. This gives us the desired running time.
    To show that, we need to show that for every $x_i\in X$, evaluating $g(x_i)$ can be done in $\BO{1}$ rounds.
    To evaluate $g(x_i)$, the algorithm checks if $G[U_i]$ contains a triangle, by computing the trace of the third power of the adjacency matrix of the subgraph $G[U_i]$.
    First, the algorithm redistributes the input, so that every vertex knows the IDs of all the vertices in $U_i$. This takes two rounds as, as follows.
    In the first round, the leader vertex informs every other vertex that it is in the set $U_i$. In the next round, each vertex that is in $U_i$ informs every other vertex that it is in the set $U_i$.
    With one additional round, the input can be redistributed in a balanced manner, as every vertex knows which entries every other vertex holds, and every vertex has to send or receive at most $\BO{n}$ entries.

    Now, we can apply \Cref{PROP:PART2} to compute the trace of the third power of the adjacency matrix of $G[U_i]$, as the input is balanced.
    We get round complexity of $\MM{\abs{U_i},\abs{U_i},\abs{U_i}}$.
    Note that for every $i\in[\ell]$, we have that $\abs{U_i}\leq 6np$ with probability $1-2^{np}$ by Chernoff's bound.
    We also get that $\MM{np,np,np}=\BO{1}$, as
    \begin{align*}
        \MM{np,np,np}\leq \BO{n^\rho\cdot p^2}=\BO{1}\;,
    \end{align*}
    where the first inequality follows from \Cref{PROP:PART2}.
    This means that the algorithm completes in $\BOs{\sqrt{\abs{X}}}$ rounds, where $\abs{X}=n^{3\rho/2}$, which gives us the desired running time.
\end{proof}

Our second algorithm which benefits from the quantum framework, is a quantum version of the simplified algorithm \algA, which gives us the following theorem.

\ThmQInCycle*

Before proving the theorem, we provide a lower bound for the probability that a uniformly random induced subgraph contains a triangle. We use this in the proof of \Cref{thm:q2}
\begin{claim}\label{claim:single t}
    Let $G$ be a graph with $n$ vertices and $t$ cycles.
    Sample a subset of vertices $U\subseteq V$, by adding each vertex from $V$ into $U$ independently with probability $p$. Let $F=G[U]$.
    If $p\geq 1/t$, then $F$ contains a triangle with probability at least $1/(3p^2)$.
\end{claim}
\begin{proof}[Proof of \Cref{claim:single t}]
    We define $t_F$ as the random variable equal to the number of triangles in $F$. We provide a lower bound on the probability that $t_F$ is positive.
    We use the second moment method: For a non-negative, integer-valued random variable $X$, we have
    \begin{align*}
        \Pr{X>0}\geq \frac{\Exp{X}^2}{\Exp{X^2}}\;.
    \end{align*}

    We have that $\Exp{t_F}=tp^3$.
    We also have that $\Exp{t_F^2}\leq tp^3 + t^2p^4(1+p+p^2)$.
    To see this, let $(C_1,\ldots,C_t)$ denote the triangles in $G$. Let $Z_i$ denote an indicator random variable for the event that $C_i$ is in $F$.
    We have that $t_F=\sum_{i=1}^t Z_i$, and therefore
    \begin{align*}
        \Exp{t_F^2}
        =\sum_{i,j\in[t]}\Exp{Z_i\cdot Z_j}
        =\sum_{i\in[t]}\Exp{Z_i\cdot Z_i} + 2\sum_{1\leq i<j\leq t}\Exp{Z_i\cdot Z_j}
        \leq tp^3 + 2t^2p^4\;.
    \end{align*}
    The first equality follows by definition.
    In the second equality, we split the sum into two products based on whether $i=j$ or $i\neq j$.
    The last inequality follows from the fact that $\Exp{Z_i\cdot Z_i}=p^3$, and $\Exp{Z_i\cdot Z_j}\leq p^4$ for $i\neq j$, as the probability that $C_i$ and $C_j$ are in $F$ is at most $p^4$, because there are at least $4$ vertices that must be sampled into $F$ for $C_i$ and $C_j$ to be in $F$.
    We therefore get that
    \begin{align*}
        \Pr{t_F>0} \geq \frac{\Exp{t_F}^2}{\Exp{t_F^2}}
        \geq \frac{(tp^3)^2}{tp^3 + 2t^2p^4}
        = \frac{tp^3}{1 + 2tp} \;.
    \end{align*}
    For $p\geq c/t$, for $c\geq 1$ we get:
        $\Pr{t_F>0}\geq \frac{c/p^2}{1 + 2c}
        \geq \frac{c/p^2}{3c}=\frac{1}{3p^2}\;.$
\end{proof}

We are ready to prove the theorem.
\begin{proof}[Proof of \Cref{thm:q2}]
    We give a quantum version of the algorithm \algA, which has two parts.
    The first part does not use the additional power of qubits, and follows the implementation of \algA:
    It samples $s=4/p^2$ subsets of vertices $\UU=(U_1,\ldots, U_{s})$ uniformly at random, where each vertex joins each set independently of all other vertices, with probability $p$. Then, the vertices are partitioned into $s$ teams, each with $n/s$ vertices.
    The $i$-th team is responsible for checking if $G[U_i]$ contains a triangle.
    As in the implementation of \algA, after sampling $\UU$, the vertices verify that $\UU$ is a $p$-balanced set. If it is not, they discard $\UU$ and resample a new set instead. If $\UU$ is a $p$-balanced set, the input is redistributed in $\BO{\log n}$ rounds (\Cref{prop:part1}).

    In the second part, which diverges from the original implementation of \algA, each team uses Grover search to find a triangle in the subgraph it holds. We explain next how the second part of the algorithm works. For that, we focus on the $i$-th team, where every other team follows the same procedure.
    The team implements the algorithm for \Cref{thm:quantum1}, and outputs \true if $G[U_i]$ contains a triangle and \false otherwise.

    To prove the correctness of the algorithm, we show that if $p\geq 1/t$ then at least one sample $G[U_i]$ contains a triangle with probability at least $1/2$. We repeat this process $8\log n$ times to amplify the success probability.
    While the value of $p$ is not available to the algorithm, we can use the same doubling approach as we did in previous sections.

    We prove that if $p\geq 1/t$ then at least one sample $G[U_i]$ contains a triangle with probability at least $1/2$.

    By \Cref{claim:single t}, for every $i\in[s]$ the probability that $G[U_i]$ contains a triangle is at least $1/3p^2$. Since the samples are independent, the probability that at least one sample contains a triangle is at least $1-(1-1/3p^2)^s\geq 1-1/e$.
    This completes the correctness proof of the algorithm: If $p\geq 1/t$, then at least one sample contains a triangle with a probability of at least $1/2$.
    If this is the case, then the $i$-th team finds it using Grover search \whp.

    In what follows, we set the parameters for Grover search and analyze the running time of the obtained algorithm.
    we focus on the $i$-th team. Define $q=(np^2)^{-\rho/2}$.
    The team leader samples $\ell=8\log n/q^3$ subsets of vertices $(W_1,\ldots, W_\ell)$ uniformly at random, where each vertex in $U_i$ is added to each set independently of all other vertices, with probability $q$.
    We emphasize that only the leader vertex is aware of the subsets.

    The search space is $X=\set{G[W_j]}_{j\in[\ell]}$. We use $x_j$ to denote $G[W_j]$ for $j\in[\ell]$.
    The Boolean function we evaluate $g:X\to\set{0,1}$ is defined as follows:
    \begin{align*}
        \forall x_j\in X\quad g(x_j)=\begin{cases}
            1 & \text{If $G[W_j]$ contains a triangle.} \\
            0 & \text{Otherwise.}
        \end{cases}
    \end{align*}
    Define $r$ as the maximal number of rounds required to evaluate $g(x)$ for some $x\in X$.
    By \cite{le2018sublinear}, the running time of the algorithm is $\BOs{r\sqrt{\abs{X}}}$.

    As part of the proof of \Cref{thm:quantum1}, we showed that if $G[U_i]$ contains a triangle, then the $i$-th team finds it using Grover search \whp.
    Unlike the proof of \Cref{thm:quantum1}, we use $np^2$ vertices to check if a subgraph with $np$ vertices contains a triangle.
    In other words, the number of vertices in the graph exceeds the number of vertices participating in the computation.
    In what follows we analyze the running time of the algorithm -- we show that $r=\BO{1}$, and get the desired running time.

    To see that $r=\BO{1}$, note that for every $x_j\in X$, evaluating $g(x_j)$ can be done in $\MM{\abs{W_j},\abs{W_j},\abs{W_j}}$ rounds, by computing the trace of the third power of the adjacency matrix of the subgraph $G[W_j]$.
    Note that for every $j\in[\ell]$, we have that $\abs{W_j}\leq 6npq$ with probability at least $1-2^{npq}$ by Chernoff's bound.

    We show that if $p=\BO{1/t}$ and $q=(np^2)^{-\rho/2}$, then in one round every team can evaluate one input.
    In each round, every team has to compute the product of two square matrices of size $\BO{npq}$, where we have $s=4/p^2$ teams.
    By \Cref{thm2:mm runtime}, which does not use the quantum power, this step takes
    \begin{align*}
        \BO{n^{\rho-2}k^2\cdot s^{1-\rho}}
        = \BO{n^{\rho}(pq)^2\cdot (p^2)^{\rho-1}}
        = \BO{(np^2)^{\rho}\cdot q^2}
    \end{align*}
    rounds.
    We plug in $p=\BO{1/t}$, and $q^2=1/(np^2)^{\rho}$, and get that we can evaluate $s$ inputs (one from every team) in $\BO{1}$ rounds.

    This means that for every $i\in[s]$, the $i$-th team can evaluate whether $G[U_i]$ contains a triangle or not, in
    \begin{align*}
        \BOs{\sqrt{\abs{X}}}
        =\BOs{\sqrt{8\log n/q^3}}
        =\TOs{((np^2)^{\rho/2})^{3/2}}
        =\TOs{(np^2)^{3\rho/4}}\;,
    \end{align*}
    rounds.
    For $p= \BOs{1/t}$, we get an $\TOs{(n/t^2)^{3\rho/4}}$-round algorithm,
    which completes the proof.
    
\end{proof}

\subsection{Challenges for \algC and larger cycles in the \QCC model}
\label{subsec:q-challenges}
Implementing \algC in the \QCC model faces the following challenge. To implement this algorithm, we need to have a better upper bound on the round complexity of computing the product of a square matrix of size $np$ with a rectangular matrix of size $np\times npq$, for $p,q\in[0,1]$.
A simple upper bound of $\MM{npq,npq,npq;1/p^2}$ can be obtained by splitting the matrices into smaller square matrices of size $npq$. However, this boils down to implementing \algA, so we can not expect to get a better running time by using this bound.

Another obstacle is a one that lies in extending our quantum algorithms for longer cycles.
To extend \Cref{thm:quantum1} to detect $h$-cycles, assuming the graph contains a single $h$-cycle, we need to sample $\Omc[1/p^h]$ random induced subgraphs of size $np$, so that at least one of them contains an $h$-cycle with probability at least $\Omc$.
The search space is then of size $1/p^{h/2}$. Each evaluation of $g(x)$ can be implemented in time $\nr\cdot p^2$.
Overall, we get an $\BO{\nr\cdot p^{2-h/2} + 1/p^{h/2}}$-round algorithm, but for $h\geq 4$, this is never faster than the classical $O(\nr)$-round algorithm.
Trying to extend \Cref{thm:q2} to detect $h$-cycles encounters the same problem as in the \CC model: We have to sample $s$ random induced subgraphs each of size $np$, where $s\gg 1/p^2$. In other words, we need too many random subgraphs compared to their size, in order to be able to apply \Cref{thm2:mm runtime}.

\section{Missing Proofs}
\label{app:proofs}
\newcommand{\emptyind}{\mathord{\ast}}

\newcommand{\fmbdual}[2]{{#1}{#2}\emptyind}
\newcommand{\fmbsingle}[1]{{#1}\emptyind}

\newcommand{\mmeaux}{\sigma} %
\newcommand{\redc}[1]{{\color{red}{#1}}}
\newcommand{\redp}{{\color{red}{p}}}
\newcommand{\redpp}{{\color{red}{p}}}
\newcommand{\rred}[1]{{\color{red}{#1}}}

\subsection{Proofs from \Cref{sec:FMM}}
\begin{proof}[Proof of \Cref{claim:balance whp}]
    We first show that \whp $\UU$ is $p$-balanced.
    The first two conditions, $a\in[0,2]$ and $p\in [1/\sqrt{n},1]$, hold by definition.
    To prove the third condition we need to show that no vertex belongs to too many sets.
    Let $X_v$ denote the number of sets in $\UU$ to which a vertex $v$ belongs.
    Note that $X_v$ is distributed as a binomial random variable, with success probability $p$ and $p^{-a}$ trials.
    A standard Chernoff bound as in \Cref{thm:chernoff} implies that
    \begin{align*}
        \Pr{X_v\geq \max\set{1,p^{1-a}}\cdot 4\log n}\leq 1/n^4.
    \end{align*}
    By a union bound, we get that with probability at least $1-1/n^3$ no vertex is in too many sets.

    To prove the fourth condition we need to show that no set is too large.
    Define $Y_{i,v}$ as the indicator random variable for the event that the vertex $v$ joined the set $U_i$.
    Define $Y_i\triangleq \sum_{v\in V}Y_{i,v}$.
    We have $\Exp{Y_i}=np\geq \sqrt{n}$, so we can again use a Chernoff bound to obtain
    \begin{align*}
        \Pr{Y_i\geq (1+3)np}\leq \exp\brak{-9np/5}\leq \exp\brak{-\sqrt{n}}\;.
    \end{align*}
    So, by a union bound, the probability that all sets are of size at most $4np$ is at least
    \begin{align*}
        1-p^{-a}\exp\brak{-\sqrt{n}}\leq 1-n\cdot\exp\brak{-\sqrt{n}}\;.
    \end{align*}
    Since conditions three and four are each satisfied \whp, their intersection is also satisfied \whp
    To summarize, the probability that $\UU$ is $p$-balanced is at least $1-\frac{1}{n^3} - \frac{n}{\exp(\sqrt{n})}\geq 1-\frac{2}{n^3}$.
\end{proof}

\subsection{Proof of \Cref{PROP:PART2}}
\label{appendix:FMM}

We provide here a proof for \Cref{PROP:PART2}, which gives our bound for FMM in the \clique model when the matrices are of dimension that is smaller or larger than $n$.
Recall that the formal definition of the task is as follows.
\DefProd*

The statement we prove is:
\SingleFMM*

The proof is a direct adaptation of the proof of the $\BO{n^{\rho}}$-round algorithm for FMM from \cite[Theorem 1]{Censor-HillelKK19}.

The high-level idea of the proof for the above is as follows, using the notation of \cite{Censor-HillelKK19}.
Take any bilinear algorithm for computing the product $P = ST$ where $S,T$ are $d \times d$
matrices using $m < d^3$ scalar multiplications.
Such an algorithm computes the following linear combinations of entries of both matrices for each $w \in [m]$:
\begin{equation*}\label{eq:bilin-1 k}
    \hat{S}^{(w)} = \sum_{(i,j) \in [d]^2}\alpha_{ijw} S_{ij}
    \hspace{10mm}\text{and}\hspace{10mm}
    \hat{T}^{(w)} = \sum_{(i,j) \in [d]^2}\beta_{ijw} T_{ij}
\end{equation*}
It then computes
$\hat{P}^{(w)} = \hat{S}^{(w)} \hat{T}^{(w)}$ for every $w \in [m]$,
and obtains $P_{ij} = \sum_{w \in [m]}\lambda_{ijw} \hat{P}^{(w)}$ for every $(i,j) \in [d]^2$,
for some given constants $\alpha_{ijw}$, $\beta_{ijw}$ and $\lambda_{ijw}$.
We show how to directly adapt this approach to matrices of dimension smaller or larger than $n$.

We take the seven-step algorithm of \cite{Censor-HillelKK19}  and generalize steps two to six, for square matrices of dimension $R\times R$ where $R$ is not necessarily equal to $n$.

For simplicity, we use the text of \cite{Censor-HillelKK19} verbatim and only insert the minor modifications that are needed. We view this result as a direct implication of  \cite{Censor-HillelKK19} rather than a new contribution.

Is in the proof of \cite{Censor-HillelKK19}, Let us fix a bilinear algorithm that computes the product of $d \times d$ matrices using $m(d)=O\left(d^\sigma\right)$ scalar multiplications for any $d$, where $2 \leq \sigma \leq 3$.
To multiply two $R \times R$ matrices on a congested clique of $n'$ nodes, fix $d$ so that $m(d)=n'$, assuming for convenience that $n'$ is such that this is possible. Note that we have $d=O\left(n'^{1 / \sigma}\right)$.

\begin{proof}[{Proof} of \Cref{PROP:PART2}]

    We assume each element in the matrices can be represented using $\BO{F}$ bits.
    \hypertarget{proof:app FMM single}{}
    The algorithm computes the matrix product $P = ST$ as follows.
    \begin{description}
        \item[Step 1: Linear combination of entries.] Each node $v$ with label $\ell'(v) = xy$ computes for $w \in V$ the linear combinations
              \begin{align*}
                  \hat{S}^{(w)}[\fmbsingle{x},\fmbsingle{y}] = & \sum_{(i,j) \in [d]^2} \alpha_{ijw} S[\fmbdual{i}{x},\fmbdual{j}{y}]\,,\hspace{10mm}\text{and} \\
                  \hat{T}^{(w)}[\fmbsingle{x},\fmbsingle{y}] = & \sum_{(i,j) \in [d]^2} \beta_{ijw} T[\fmbdual{i}{x},\fmbdual{j}{y}]\,.
              \end{align*}
              The computation is performed entirely locally.
        \item[Step 2: Distributing the linear combinations.] Each node $v$ with label $\ell'(v) = xy$ sends, for $w \in V$, the submatrices $\hat{S}^{(w)}[\fmbsingle{x},\fmbsingle{y}]$ and $\hat{T}^{(w)}[\fmbsingle{x},\fmbsingle{y}]$ to node $w$.
              Each submatrix has
              $(\redc{\frac{R}{d}\cdot\frac{1}{\sqrt{n'}}})^2 = \redc{\frac{R^2}{d^2\cdot n'}}$
              entries and there are $\redc{n'}$ recipients each receiving two submatrices, for a total of $O((\redc{R/d})^2)$ messages per node.

              Dually, each node $w \in V$ receives the submatrices $\hat{S}^{(w)}[\fmbsingle{x},\fmbsingle{y}]$ and $\hat{T}^{(w)}[\fmbsingle{x},\fmbsingle{y}]$ from node $v \in V$ with label $\ell'(v) = xy$. Node $u$ now has the matrices $\hat{S}^{(w)}$ and $\hat{T}^{(w)}$. The total number of received messages is $O((\redc{R/d})^2)$ per node.
        \item[Step 3: Multiplication.] Node $w \in V$ computes the product $\hat{P}^{(w)} = \hat{S}^{(w)}\hat{T}^{(w)}$. The computation is performed entirely locally.
        \item[Step 4: Distributing the products.] Each node $w$ sends, for $x, y \in [(\redc{n'})^{1/2}]$, the submatrix $\hat{P}^{(w)}[\fmbsingle{x},\fmbsingle{y}]$ to node $v \in V$ with label $xy$.
              Each submatrix has $((\redc{R/d})^2\cdot 1/\redc{n'})$ entries and there are $\redc{n'}$ recipients, for a total of $O((\redc{R/d})^2)$ messages sent by each node.

              Dually, each node $v \in V$ with label $\ell'(v) = xy$ receives the submatrix $\hat{P}^{(w)}[\fmbsingle{x},\fmbsingle{y}]$ from each node $w \in V$.
              The total number of received messages is $O((\redc{R/d})^2)$ per node.
        \item[Step 5: Linear combination of products.] Each node $v \in V$ with label $\ell'(v) = xy$ computes for $i,j \in [d]$ the linear combination
              \[ P[\fmbdual{i}{x},\fmbdual{j}{y}] = \sum_{w \in V} \lambda_{ijw} \hat{P}^{(w)}[\fmbsingle{x},\fmbsingle{y}]\,.\]
              Node $v \in V$ now has the submatrix $P[\fmbdual{\emptyind}{x},\fmbdual{\emptyind}{y}]$. The computation is performed entirely locally.
    \end{description}
    
    ~\\\textbf{Analysis.}
    The maximal number of entries sent or received by a node in the above steps is $O((\redc{R/d})^2)$, where each entry can be represented using $\BO{F}$ bits.  First, let's assume each vertex can send in each round, one message of $B$ bits to every other vertex in $V$.
    We get that sending $O((\redc{R/d})^2)$ messages takes
    \begin{align*}
        (R/d)^2 / n' \cdot \frac{F}{B}
        = R^2/(n')^{1+2/\mmeaux} \cdot \frac{F}{B}
        = (R/n')^2 \cdot (n')^{1-2/\mmeaux} \cdot \frac{F}{B}
        =\BO{(n')^{\rho}\cdot (R/n')^2 \cdot \frac{F}{B}}
    \end{align*}
    rounds.
    We can further see, that the communication pattern is fixed. In other words, the content of the messages does not effect the communication pattern. Therefore, increasing the bandwidth by a factor of $C$ will decrease the number of rounds the algorithm takes by a factor of $C$.
    We therefore get the desired running time. This completes the proof of \Cref{PROP:PART2}.
\end{proof}

\subsection{Proof of \Cref{CLAIM:NUMERIC RHO}}
\label{appendix:convexity}
This section is devoted to proving some analytic bound on $\rho(z)$.
Specifically, we prove \Cref{CLAIM:NUMERIC RHO}, building upon
\cite{le2016further,gall2018improved,alman2021refined}.

\ClaimNumericRho*

Current bounds on $\rho(z)$ are of the following.
\begin{theorem}[{\cite[Proposition 3]{le2016further}}]
    For $z\in[1/2,1]$
    \begin{align*}
        \RM{n^z}\eqdef\BO{n^{\rho(z)}}\leq \BO{n^{1-2/\omega(\gamma)}}
    \end{align*}
    where $\gamma$ is the solution for the following equation
    \begin{align}
        \gamma=1-(z-1)\omega(\gamma)\;.\label{eq:gy}
    \end{align}
\end{theorem}
We do not know how to solve or even give an approximation for $\gamma$ that solves \Cref{eq:gy}.
If $\rho(z)$ was convex, we could get the following bound:
\begin{align*}
    \BO{n^{\rho(z)}}\leq \BO{n^{\frac{\rho(1)\cdot (z-\bz)}{1-\bz}}}
\end{align*}
where $\bz$ was defined in \Cref{def:rect}.
One could hope to generalize the proof of \cite[Lemma 3.6]{jiang2020faster} that $\omega(z)$ is convex to show that $\rho(z)$ is also convex.
This requires a deeper study of multiple multiplications of rectangular matrices in the \clique model. Instead, we bypass this by solving our problem numerically, as follows.

\newcommand{\wg}{\omega(\gamma)}
\newcommand{\Lo}[1]{L_{\omega}(#1)}
\newcommand{\gi}[1][i]{\gamma_{#1}}
\newcommand{\RR}[1]{\mathsf{S}\brak{#1}}

We take a set of values of $(\gamma,\wg)$ given in \cite{gall2018improved,alman2021refined},
and create a piece-wise linear function of $\omega(z)$, as suggested by \cite{Complexity}.
We denote this linear approximation by $\Lo{z}$.

Using $\Lo{z}$, we numerically find a set of points $\set{\gi}_{i\in[r+1]}$, and set of values $\set{y_i}_{i\in[r+1]}$
where $y_i$ is an upper bound on $\omega(\gi)$.
We use those points and values to construct a step function  $\RR{x}$ which will satisfy
\begin{align*}
    \forall i\in[r]\quad \RR{\gi[i+1]}\leq \r{\gi}\;.
\end{align*}
Using the fact that $\rho$ is monotone non-decreasing, we get that
for any $x$,
\begin{align*}
    \rho(z)\leq \RR{z}\leq \r{z}
\end{align*}
which proves \Cref{CLAIM:NUMERIC RHO}.

The set of points $\gi$ and the matching values are reported in \Cref{table:values1,table:values2}, and in \Cref{fig:convex}.

\begin{table}
    \begin{center}
        \begin{tabular}{lll|ll|l}
            $i$    & $\gamma$     & $\omega(\gamma)$ & $y\triangleq 1-\frac{1-\gamma}{\omega(\gamma)}$ & $\RR{y}\leq 1-\frac{2}{\omega(\gamma)}$ & $\r{y}$      \\
            \toprule
            $0   $ & $0.32132067$ & $2.00000000$     & $0.66066033$                                    & $0.00000000$                            & $0.00000000$ \\
            $1   $ & $0.32132087$ & $2.00000000$     & $0.66066043$                                    & $0.00000000$                            & $0.00000005$ \\
            $2   $ & $0.32133367$ & $2.00000000$     & $0.66066683$                                    & $0.00000000$                            & $0.00000300$ \\
            $3   $ & $0.32174167$ & $2.00000470$     & $0.66087163$                                    & $0.00000235$                            & $0.00009756$ \\
            $4   $ & $0.32826999$ & $2.00008004$     & $0.66414843$                                    & $0.00004002$                            & $0.00161048$ \\
            $5   $ & $0.35386722$ & $2.00199552$     & $0.67725563$                                    & $0.00099677$                            & $0.00766215$ \\
            $6   $ & $0.40379059$ & $2.01062109$     & $0.70347003$                                    & $0.00528249$                            & $0.01976549$ \\
            $7   $ & $0.45266218$ & $2.02481060$     & $0.72968443$                                    & $0.01225329$                            & $0.03186884$ \\
            $8   $ & $0.50117496$ & $2.04351766$     & $0.75589883$                                    & $0.02129546$                            & $0.04397218$ \\
            $9   $ & $0.54983423$ & $2.06605377$     & $0.78211323$                                    & $0.03197098$                            & $0.05607553$ \\
            $10  $ & $0.59900185$ & $2.09210204$     & $0.80832763$                                    & $0.04402368$                            & $0.06817887$ \\
            $11  $ & $0.64903513$ & $2.12117239$     & $0.83454203$                                    & $0.05712520$                            & $0.08028221$ \\
            $12  $ & $0.70018485$ & $2.15317060$     & $0.86075643$                                    & $0.07113723$                            & $0.09238556$ \\
            $13  $ & $0.75268386$ & $2.18807362$     & $0.88697083$                                    & $0.08595397$                            & $0.10448890$ \\
            $14  $ & $0.77950259$ & $2.20669601$     & $0.90007803$                                    & $0.09366764$                            & $0.11054058$ \\
            $15  $ & $0.80676694$ & $2.22580864$     & $0.91318523$                                    & $0.10145016$                            & $0.11659225$ \\
            $16  $ & $0.83446870$ & $2.24578438$     & $0.92629243$                                    & $0.10944255$                            & $0.12264392$ \\
            $17  $ & $0.86265499$ & $2.26640564$     & $0.93939963$                                    & $0.11754544$                            & $0.12869559$ \\
            $18  $ & $0.89134672$ & $2.28776663$     & $0.95250683$                                    & $0.12578496$                            & $0.13474726$ \\
            $19  $ & $0.90589165$ & $2.29871387$     & $0.95906043$                                    & $0.12994826$                            & $0.13777310$ \\
            $20  $ & $0.92057056$ & $2.30993766$     & $0.96561403$                                    & $0.13417577$                            & $0.14079894$ \\
            $21  $ & $0.93539352$ & $2.32127159$     & $0.97216763$                                    & $0.13840328$                            & $0.14382477$ \\
            $22  $ & $0.94285968$ & $2.32698037$     & $0.97544443$                                    & $0.14051703$                            & $0.14533769$ \\
            $23  $ & $0.95036252$ & $2.33272358$     & $0.97872123$                                    & $0.14263309$                            & $0.14685061$ \\
            $24  $ & $0.95790024$ & $2.33861989$     & $0.98199803$                                    & $0.14479475$                            & $0.14836353$ \\
            $25  $ & $0.96547617$ & $2.34454608$     & $0.98527483$                                    & $0.14695641$                            & $0.14987644$ \\
            $26  $ & $0.97309059$ & $2.35050238$     & $0.98855163$                                    & $0.14911807$                            & $0.15138936$ \\
            $27  $ & $0.97691232$ & $2.35349190$     & $0.99019003$                                    & $0.15019890$                            & $0.15214582$ \\
            $28  $ & $0.98074379$ & $2.35648902$     & $0.99182843$                                    & $0.15127973$                            & $0.15290228$ \\
            $29  $ & $0.98458503$ & $2.35949380$     & $0.99346683$                                    & $0.15236056$                            & $0.15365874$ \\
            $30  $ & $0.98650933$ & $2.36099906$     & $0.99428603$                                    & $0.15290098$                            & $0.15403697$ \\
            $31  $ & $0.98843608$ & $2.36250624$     & $0.99510523$                                    & $0.15344139$                            & $0.15441520$ \\
            $32  $ & $0.99036530$ & $2.36401535$     & $0.99592443$                                    & $0.15398180$                            & $0.15479343$ \\
            $33  $ & $0.99229698$ & $2.36552639$     & $0.99674363$                                    & $0.15452222$                            & $0.15517166$ \\
            $34  $ & $0.99326375$ & $2.36628263$     & $0.99715323$                                    & $0.15479243$                            & $0.15536077$ \\
            $35  $ & $0.99423113$ & $2.36703936$     & $0.99756283$                                    & $0.15506263$                            & $0.15554989$ \\
            $36  $ & $0.99519913$ & $2.36779657$     & $0.99797243$                                    & $0.15533284$                            & $0.15573900$ \\
            $37  $ & $0.99616776$ & $2.36855427$     & $0.99838203$                                    & $0.15560305$                            & $0.15592812$ \\
            $38  $ & $0.99665230$ & $2.36893330$     & $0.99858683$                                    & $0.15573815$                            & $0.15602267$ \\
            $39  $ & $0.99713700$ & $2.36931245$     & $0.99879163$                                    & $0.15587326$                            & $0.15611723$ \\
            $40  $ & $0.99762186$ & $2.36969172$     & $0.99899643$                                    & $0.15600836$                            & $0.15621179$ \\
            $41  $ & $0.99810687$ & $2.37007112$     & $0.99920123$                                    & $0.15614346$                            & $0.15630635$ \\
            $42  $ & $0.99834943$ & $2.37026086$     & $0.99930363$                                    & $0.15621102$                            & $0.15635362$ \\
            $43  $ & $0.99859203$ & $2.37045063$     & $0.99940603$                                    & $0.15627857$                            & $0.15640090$ \\
            $44  $ & $0.99883467$ & $2.37064043$     & $0.99950843$                                    & $0.15634612$                            & $0.15644818$ \\
            $45  $ & $0.99907735$ & $2.37083027$     & $0.99961083$                                    & $0.15641367$                            & $0.15649546$ \\
            \bottomrule
        \end{tabular}
        \caption{Values of $\omega,\; \gamma$ and more.
        }
        \label{table:values1}
    \end{center}
\end{table}
\begin{table}
    \begin{center}
        \begin{tabular}{lll|ll|l}
            $i$    & $\gamma$     & $\omega(\gamma)$ & $y\triangleq 1-\frac{1-\gamma}{\omega(\gamma)}$ & $\RR{y}\leq 1-\frac{2}{\omega(\gamma)}$ & $\r{y}$      \\
            \toprule
            $46  $ & $0.99919871$ & $2.37092520$     & $0.99966203$                                    & $0.15644745$                            & $0.15651910$ \\
            $47  $ & $0.99932007$ & $2.37102013$     & $0.99971323$                                    & $0.15648122$                            & $0.15654274$ \\
            $48  $ & $0.99944144$ & $2.37111508$     & $0.99976443$                                    & $0.15651500$                            & $0.15656638$ \\
            $49  $ & $0.99956283$ & $2.37121003$     & $0.99981563$                                    & $0.15654878$                            & $0.15659002$ \\
            $50  $ & $0.99962352$ & $2.37125750$     & $0.99984123$                                    & $0.15656566$                            & $0.15660184$ \\
            $51  $ & $0.99968422$ & $2.37130499$     & $0.99986683$                                    & $0.15658255$                            & $0.15661366$ \\
            $52  $ & $0.99974492$ & $2.37135247$     & $0.99989243$                                    & $0.15659944$                            & $0.15662548$ \\
            $53  $ & $0.99980562$ & $2.37139995$     & $0.99991803$                                    & $0.15661633$                            & $0.15663730$ \\
            $54  $ & $0.99983598$ & $2.37142369$     & $0.99993083$                                    & $0.15662477$                            & $0.15664321$ \\
            $55  $ & $0.99986633$ & $2.37144744$     & $0.99994363$                                    & $0.15663322$                            & $0.15664912$ \\
            $56  $ & $0.99989668$ & $2.37147118$     & $0.99995643$                                    & $0.15664166$                            & $0.15665503$ \\
            $57  $ & $0.99992704$ & $2.37149493$     & $0.99996923$                                    & $0.15665010$                            & $0.15666094$ \\
            $58  $ & $0.99994221$ & $2.37150680$     & $0.99997563$                                    & $0.15665433$                            & $0.15666389$ \\
            $59  $ & $0.99995739$ & $2.37151867$     & $0.99998203$                                    & $0.15665855$                            & $0.15666685$ \\
            $60  $ & $0.99997257$ & $2.37153054$     & $0.99998843$                                    & $0.15666277$                            & $0.15666980$ \\
            $61  $ & $0.99998775$ & $2.37154242$     & $0.99999483$                                    & $0.15666699$                            & $0.15667276$ \\
            $62  $ & $1.00000000$ & $2.37155200$     & $1.00000000$                                    & $0.15667040$                            & $0.15667514$ \\
            \bottomrule
        \end{tabular}
        \caption{Values of $\omega,\; \gamma$ and more.
        }
        \label{table:values2}
    \end{center}
\end{table}
\begin{figure*}[t]
    \centering
    \includegraphics[width=\textwidth]{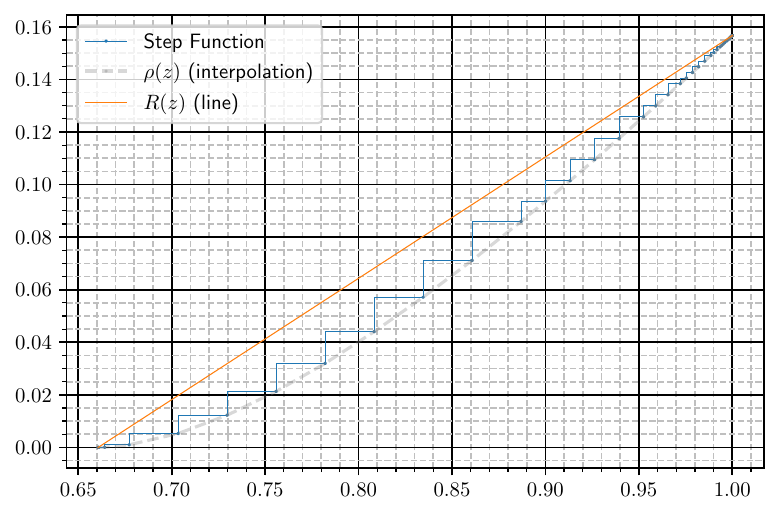}
    \caption{An illustration of the line $R(z)$ that bounds the step function we use in order to bound $\rho(z)$, and an image of an interpolation of what could be $\rho(z)$.}\label{fig:convex}
\end{figure*}
\begin{figure*}[t]
    \centering
    \includegraphics[width=\textwidth]{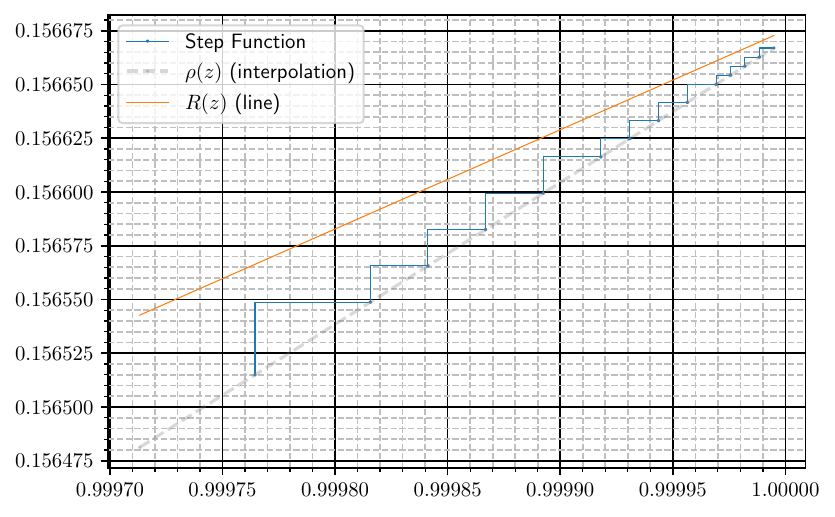}
    \caption{Enlarged view of the plot around $z=1$ from the previous figure.}\label{fig:convex-zoom}
\end{figure*}
\clearpage

\bibliographystyle{alpha}
\addcontentsline{toc}{section}{Bibliography}
\bibliography{refs.bib}

\end{document}